\documentclass[a4paper,USenglish,cleveref, autoref, thm-restate]{lipics-v2021}
\pdfoutput=1
\usepackage{microtype}
\usepackage[utf8]{inputenc}
\usepackage{todonotes}
\usepackage{mathtools}
\usepackage{hyperref}
\hypersetup{colorlinks=true,citecolor=blue}
\usetikzlibrary{automata}

\newcommand{\IN}{\mathbb{N}}
\newcommand{\IR}{\mathbb{R}}
\newcommand{\automaton}{\mathcal{A}}
\newcommand{\game}{\mathcal{G}}
\newcommand{\player}{\mathcal{P}}
\newcommand{\paths}{\mathsf{Paths}}
\newcommand{\plays}{\mathsf{Plays}}
\newcommand{\last}{\mathsf{last}}
\newcommand{\ex}{\mathsf{Ex}}
\newcommand{\pr}{\mathsf{Pr}}
\newcommand{\cobuchi}{\mathsf{coB\ddot{u}chi}}
\newcommand{\buchi}{\mathsf{B\ddot{u}chi}}
\newcommand{\safe}{\mathsf{Safe}}
\newcommand{\reach}{\mathsf{Reach}}
\newcommand{\wc}{\mathsf{WC}}
\newcommand{\outcome}{\mathsf{Outcome}}
\newcommand{\true}{\mathsf{true}}
\newcommand{\delay}{\mathsf{delay}}
\newcommand{\bad}{\mathsf{bad}}
\newcommand{\badset}{\mathsf{Bad}}
\newcommand{\init}{\mathsf{init}}

\title{Time Flies When Looking out of the Window: Timed Games with Window Parity Objectives}

\titlerunning{Timed Games with Window Parity Objectives}

\author{James C.~A.~Main}{UMONS -- Université de Mons, Belgium}{}{}{}

\author{Mickael Randour}{F.R.S.-FNRS \& UMONS -- Université de Mons, Belgium}{}{}{F.R.S.-FNRS Research Associate.}

\author{Jeremy Sproston}{Universit\`{a} degli Studi di Torino, Italy}{}{}{}

\authorrunning{J.~C.~A.~Main, M.~Randour and J.~Sproston}

\Copyright{James C.~A.~Main, Mickael Randour, and Jeremy Sproston}



\category{}

\relatedversion{}

\supplement{}

\funding{Research supported by F.R.S.-FNRS under Grant n° F.4520.18 (ManySynth).}

\nolinenumbers

\hideLIPIcs

\begin{document}

\maketitle

\begin{abstract}
  The \textit{window mechanism} was introduced by Chatterjee et al.~to reinforce
  mean-payoff and total-payoff objectives with time bounds in two-player
  turn-based games on graphs~\cite{Chatterjee0RR15}. It has since proved useful
  in a variety of settings, including parity objectives
  in games~\cite{BruyereHR16} and both mean-payoff and parity objectives
  in Markov decision processes~\cite{DBLP:journals/lmcs/BrihayeDOR20}. 
  
  We study \textit{window parity} objectives in timed automata and
  timed games: given a bound on the window size, a path satisfies such an
  objective if, in all states along the path, we see a sufficiently small
  window in which the smallest priority is even.
  We show that checking that all time-divergent paths of
  a timed automaton satisfy such a window parity objective can be done in
  polynomial space, and that the corresponding timed games can be solved in exponential time. This matches the complexity class of timed parity games, while adding the ability to reason about time bounds. We also consider multi-dimensional objectives and show that
  the complexity class does not increase. To the best of our knowledge, this is the first study of the window mechanism in a real-time setting.
\end{abstract}

\section{Introduction}
\subparagraph{Timed automata and games.} \textit{Timed automata}~\cite{AlurD94}
are extensions of finite automata with real-valued variables called \textit{clocks}.
Clocks increase at the same rate and measure the elapse of time between actions.
Transitions are constrained by the values of clocks, and clocks can be reset on transitions.
Timed automata are used to model real-time systems~\cite{DBLP:books/daglib/0020348}. Not all paths of timed
automata are meaningful; infinite paths that take a finite amount of time,
called \textit{time-convergent} paths, are often disregarded when checking properties of
timed automata. Timed automata induce uncountable transition
systems. However, many properties can be checked using the region abstraction,
which is a finite quotient of the transition system.

Timed automaton games~\cite{MalerPS95}, or simply \textit{timed games}, are games played
on timed automata: one player represents the system and the other its
environment. Players play an infinite amount of rounds: for each round,
both players simultaneously present a delay and an action, and the play proceeds according to the fastest move (note that we use paths for automata and plays for games to refer to sequences of consecutive states and transitions). When defining winning conditions for
players, convergent plays must be taken in account; we must not allow a player
to achieve its objective by forcing convergence but cannot either require a
player to force divergence (as it also depends on its opponent).
Given an objective as a set of plays, following~\cite{AlfaroFHMS03}, we
declare a play winning for a player if either it is \textit{time-divergent and
belongs to the objective}, or it is \textit{time-convergent and the player is not
responsible for convergence}.

\subparagraph{Parity conditions.} The
class of $\omega$-regular specifications is widely used (e.g., it can express liveness and safety), and parity conditions are a canonical
way of representing them.
In (timed) parity games, locations are labeled with a non-negative integer
priority and the parity objective is to ensure the smallest priority occurring
infinitely often along the path/play is even.
Timed games with $\omega$-regular objectives given as parity automata are
shown to be solvable in~\cite{AlfaroFHMS03}. Furthermore, a reduction from timed parity games to
classical turn-based parity games on a graph is established in~\cite{ChatterjeeHP11}.

\subparagraph{Real-timed windows.} The parity objective
can be reformulated: for all odd priorities seen infinitely often,
a smaller even priority must be seen infinitely often. One can see
the odd priority as a request and the even one as a response.
The parity objective does not specify any \textit{timing constraints}
between requests and responses. In applications however, this may not be
sufficient: for example, a server should respond to requests in a timely
manner.

We revisit the window mechanism introduced by Chatterjee
et al.~for mean-payoff and total-payoff games~\cite{Chatterjee0RR15}
and later applied to parity games~\cite{BruyereHR16} and to parity and
mean-payoff objectives in Markov decision processes~\cite{DBLP:journals/lmcs/BrihayeDOR20}: we provide \textit{the first} (to the best of our knowledge) \textit{study of window objectives in the real-time setting}. More precisely, we lift the (resp.~direct) fixed window parity objective of~\cite{BruyereHR16} to its real-time counterpart, the \emph{(resp.~direct) timed window parity objective}, and study it in timed automata and games.

Intuitively, given a non-negative integer bound $\lambda$ on the window
size, the \textit{direct} timed window parity objective requires that
\textit{at all}
times along a path/play, we see a window of size at most $\lambda$ such
that the smallest priority in this window is even. While time was counted as
steps in prior works (all in a discrete setting), \textit{we naturally measure window size using delays} between configurations
in real-time models. The \textit{(non-direct)} timed window parity objective is simply a prefix-independent version of the direct one, thus more closely matching the spirit of classical parity: it asks that some suffix satisfies the direct objective.

\subparagraph{Contributions.} We extend window parity objectives to a dense-time
setting, and study both \textit{verification} of timed automata and
\textit{realizability} in timed games. We consider
adaptations of the \emph{fixed window parity objectives} of~\cite{BruyereHR16},
where the window size is given as a parameter. We establish that (a) verifying that all time-divergent paths of a timed
automaton satisfy a timed window parity specification is \textsf{PSPACE}-complete; and that (b) checking the existence of a winning strategy for a window
parity objective in timed games is \textsf{EXPTIME}-complete. These results (Theorem~\ref{theorem:complexity}) hold for both the direct and prefix-independent variants, and they extend to multi-dimensional objectives, i.e., conjunctions of window parity.

All algorithms are based on a reduction to an \textit{expanded timed automaton} (Definition~\ref{definition:singledim:extended_automaton}).
We establish that, similarly to the discrete case, it suffices to keep track of
one window at a time (or one per objective in the multi-dimensional case) instead of all currently open windows, thanks to the so-called \textit{inductive property of windows} (Lemma~\ref{lemma:inductive_prop}).
A window can be summarized using its smallest priority and
its current size: we encode the priorities in a window by extending
locations with priorities and using an additional clock to measure the window's size.
The (resp.~direct) timed window parity objective translates to a co-Büchi (resp.~safety)
objective on the expanded automaton. Locations to avoid for the co-Büchi
(resp.~safety) objective indicate
a window exceeding the supplied bound without the smallest priority of the window being
even --- a \textit{bad window}.
To check that all time-divergent paths of the expanded automaton satisfy the
safety (resp.~co-Büchi) objective, we check for the existence of a
time-divergent path visiting (resp.~infinitely often)
an unsafe location using the \textsf{PSPACE} algorithm of \cite{AlurD94}.
To solve the similarly-constructed
expanded game, we use the \textsf{EXPTIME} algorithm of~\cite{AlfaroFHMS03}.

Lower bounds (Lemma~\ref{lemma:complexity:lowerbound}) are established by encoding safety objectives on timed automata as (resp.~direct) timed window parity
objectives.
Checking safety properties over time-divergent paths in timed automata is
\textsf{PSPACE}-complete~\cite{AlurD94} and solving safety timed games
is \textsf{EXPTIME}-complete~\cite{HenzingerK99}.
\subparagraph{Comparison.} Window variants constitute \textit{conservative approximations} of classical objectives  (e.g.,~\cite{Chatterjee0RR15,BruyereHR16,DBLP:journals/lmcs/BrihayeDOR20}), strengthening them by enforcing timing constraints. Complexity-wise, the situation is varied. In one-dimension turn-based games on graphs, window variants~\cite{Chatterjee0RR15,BruyereHR16} provide \textit{polynomial-time alternatives} to the classical objectives, bypassing long-standing complexity barriers. However, in multi-dimension games, their complexity becomes worse than for the original objectives: in particular, fixed window parity games are \textsf{EXPTIME}-complete for multiple dimensions~\cite{BruyereHR16}.
We show that timed games with multi-dimensional window parity objectives are in
the same complexity class as untimed ones, i.e., dense time comes for free.

For classical parity objectives, timed games can be solved in exponential
time~\cite{ChatterjeeHP11, AlfaroFHMS03}. The approach of~\cite{ChatterjeeHP11} is as follows: from a timed parity game, one builds a corresponding turn-based parity game on a graph, the construction being polynomial
in the number of priorities and the size of the region abstraction.
We recall that despite recent progress on quasi-polynomial-time algorithms (starting with~\cite{CaludeJKL017}), no polynomial-time algorithm is known for parity games; the blow-up comes from the number of priorities. Overall, \textit{the two sources of blow-up} --- region
abstraction and number of priorities --- \textit{combine in a single-exponential solution} for
timed parity games. 
We establish that (multi-dimensional) window parity games can be solved in time polynomial in the size of the region
abstraction, the number of priorities and the window
size, and exponential in the number of dimensions. Thus \textit{even for conjunctions of objectives, we match the complexity class of single parity objectives of timed games, while avoiding the blow-up related to the number of priorities and enforcing time bounds between odd priorities and smaller even priorities via the window mechanism}.

\subparagraph{Outline.} This work is organized as follows. Section~\ref{section:prelim} summarizes all prerequisite notions and vocabulary.
In Section~\ref{section:windows}, we introduce the timed window parity objective,
compare it to the classical parity objective, and establish some useful
properties.
Our reduction from window parity objectives to safety or co-Büchi objectives is presented in Section~\ref{section:reduction}: the construction of the expanded timed automaton/game used in the reduction is provided in Section~\ref{section:reduction:construction}, Section~\ref{section:reduction:mappings} develops mappings between plays of a game and plays of its expansion such that a time-divergent play in one game satisfies the objective if and only if its image satisfies the objective of the other game and Section~\ref{section:reduction:translations} shows how these mappings can be used to transfer winning strategies between a
timed game and its expansion.
The reduction is extended to the case of multiple
window parity objectives in Section~\ref{section:multi}.
Finally, Section~\ref{section:complexity} presents the complexity.

This paper is a full version of a preceding conference version~\cite{MRS21}. This version presents in full details the contributions of the conference version, with detailed proofs.

\subparagraph{Related work.} In addition to the aforementioned foundational works, the window mechanism has seen diverse extensions and applications: e.g.,~\cite{DBLP:conf/csl/BaierKKW14,DBLP:conf/rp/Baier15,DBLP:conf/concur/BrazdilFKN16,DBLP:conf/concur/BruyereHR16,DBLP:journals/acta/HunterPR18,DBLP:conf/fsttcs/0001PR18,DBLP:conf/fsttcs/BordaisGR19}. Window parity games are strongly linked to the concept of \textit{finitary $\omega$-regular games}: see, e.g.,~\cite{DBLP:journals/tocl/ChatterjeeHH09}, or~\cite{BruyereHR16} for a complete list of references. The window mechanism can be used to ensure a certain form of (local) guarantee over paths: different techniques have been considered in stochastic models, notably variance-based~\cite{DBLP:journals/jcss/BrazdilCFK17} or worst-case-based~\cite{DBLP:journals/iandc/BruyereFRR17,DBLP:conf/icalp/BerthonRR17} methods. Finally, let us recall that game models provide a useful framework for controller synthesis~\cite{rECCS}, and that timed automata have been extended in a number of directions (see, e.g.,~\cite{DBLP:journals/corr/abs-2001-04347} and references therein): applications of the window mechanism in such richer models could be of interest.

\section{Preliminaries}\label{section:prelim}
\subparagraph*{Timed automata.}
A clock variable, or \textit{clock}, is a real-valued variable.
Let $C$ be a set of clocks. A \emph{clock constraint} over
$C$ is a conjunction of formulae of the form $x \sim c$ with
$x\in C$, $c\in \IN$, and
$\sim\in \{\leq,\geq, >, <\}$.  We write $x=c$ as shorthand for the clock
constraint $x\geq c\land x\leq c$. Let $\Phi(C)$ denote the set
of clock constraints over $C$.

Let $\IR_{\geq 0}$ denote the set of non-negative real numbers.
We refer to functions $\nu\in \IR_{\geq 0}^C$ as
\emph{clock valuations} over $C$. A clock valuation $\nu$ over a set
$C$ of clocks satisfies a clock constraint of the form  $x\sim c$ if
$\nu(x) \sim c$ and a conjunction $g\land h$ if it satisfies both
$g$ and $h$. Given a clock constraint $g$ and clock valuation $\nu$, we
write $\nu\models g$ if $\nu$ satisfies $g$.

For a clock valuation $\nu$ and $d\geq 0$, we let
$\nu +d$ be the valuation defined by $(\nu + d)(x) = \nu(x) +d $ for
all $x\in C$. For any valuation $\nu$ and $D\subseteq C$,
we define $\mathsf{reset}_D(\nu)$ to be the valuation agreeing with
$\nu$ for clocks in $C\setminus D$ and that assigns $0$ to clocks in $D$.
We denote by $\mathbf{0}^C$ the zero valuation, assigning 0 to all clocks in
$C$.

A \emph{timed automaton} (TA) is a tuple $(L, \ell_\init, C, \Sigma, I, E)$ where
$L$ is a finite set of \emph{locations}, $\ell_\init\in L$ is an initial
location, $C$ a finite set of \emph{clocks} containing a special
clock $\gamma$ which keeps track of the total time elapsed,
$\Sigma$ a finite set of actions,
$I\colon L\to \Phi(C)$ an \emph{invariant} assignment function and
$E\subseteq L\times \Phi(C)\times \Sigma\times 2^{C\setminus\{\gamma\}}\times L$
an edge relation.
We only consider deterministic timed automata, i.e., we assume that
in any location $\ell$, there are no two outgoing edges
$(\ell, g_1, a, D_1, \ell_1)$ and $(\ell, g_2, a, D_2, \ell_2)$
sharing the same action
such that the conjunction $g_1\land g_2$ is satisfiable. For an edge
$(\ell, g, a, D, \ell')$, the clock constraint $g$ is called the
\textit{guard} of the edge.

A TA $\automaton = (L, \ell_\init, C, \Sigma, I, E)$
gives rise to an uncountable transition system
$\mathcal{T}(\automaton) = (S, s_\init, M, \to)$
with the state space $S = L\times \IR_{\geq 0}^C$,
the initial state
$s_\init= (\ell_\init, \mathbf{0}^C)$, set of actions $M = \IR_{\geq 0}\times
(\Sigma\cup\{\bot\})$ and the
transition relation $\to\subseteq S\times M\times S$ defined as follows:
for any action $a\in\Sigma$ and
delay $d\geq 0$, we have that $((\ell, \nu), (d, a), (\ell', \nu'))\in\to$ if
and only if there is some edge $(\ell, g, a, D, \ell')\in E$
such that $\nu + d\models g$, $\nu' = \mathsf{reset}_{D}(\nu +d)$,
$\nu +d\models I(\ell)$ and $\nu'\models I(\ell')$; for any delay $d\geq 0$,
$((\ell, \nu)(d, \bot),(\ell, \nu +d))\in\to$ if and only if $\nu + d\models I(\ell)$.
Let us note that the satisfaction set of clock constraints is convex: it
is described by a conjunction of inequalities. Whenever
$\nu\models I(\ell)$, the above conditions $\nu + d\models I(\ell)$
(the invariant holds after the delay)
are equivalent to requiring $\nu + d'\models I(\ell)$ for all $0\leq d'\leq d$
(the invariant holds at each intermediate time step).

A \textit{move} is any pair in $\IR_{\geq 0}\times (\Sigma\cup\{\bot\})$
(i.e., an action in the transition system).
For any move $m=(d, a)$ and states $s$, $s'\in S$, we write
$s\xrightarrow{m}s'$ or $s\xrightarrow{d, a}s'$ as shorthand for
$(s,m,s')\in\to$.
Moves of the form $(d, \bot)$ are called \textit{delay moves}.
We say a move $m$ is
enabled in a state $s$ if there is some $s'$ such that $s\xrightarrow{m}s'$.
There is at most one successor per move in a state, as we do not allow two
guards on edges labeled by the same action to be simultaneously satisfied.

A \textit{path} in a TA $\automaton$ is a finite or infinite
sequence
$s_0(d_0, a_0)s_1\ldots\in S(MS)^*\cup (SM)^\omega$ such that for all $j$, $s_j$ is a state of $\mathcal{T}(\automaton)$ and
for all $j > 0$, $s_{j-1}\xrightarrow{d_{j-1}, a_{j-1}}s_{j}$ is a
transition in $\mathcal{T}(\automaton)$. A path is \textit{initial}
if $s_0=s_{\init}$. For clarity, we write
$s_0\xrightarrow{d_0, a_0}s_1\xrightarrow{d_1, a_1}\cdots$ instead
of $s_0(d_0, a_0)s_1(d_1, a_1)\ldots$.

An infinite path
$\pi = (\ell_0, \nu_0)\xrightarrow{d_0, a_0}(\ell_1, \nu_1)\ldots$
is \textit{time-divergent} if the sequence $(\nu_j(\gamma))_{j\in\IN}$
is not bounded from above. A path which is not time-divergent is called
\textit{time-convergent}; time-convergent paths are traditionally
ignored in analysis
of timed automata \cite{DBLP:journals/iandc/AlurCD93, AlurD94} as they model unrealistic behavior.
This includes ignoring
\emph{Zeno paths}, which are time-convergent paths along which infinitely
many actions appear.
We write $\paths(\automaton)$ for the set of paths of $\automaton$.

\subparagraph*{Priorities.} A \textit{priority function} is a function
$p\colon L\to \{0, \ldots, d-1\}$ with $d\leq |L|$.  We use priority functions
to express parity objectives. A \textit{$k$-dimensional priority function}
is a function $p\colon L\to \{0, \ldots, d-1\}^k$ which assigns vectors
of priorities to locations.

\subparagraph*{Timed games.}
We consider two player games played on TAs. We refer to the players
as player 1 ($\player_1$) for the system and player 2 ($\player_2$) for
the environment.
We use the notion of timed automaton games of \cite{AlfaroFHMS03}.

A \textit{timed} (automaton) \textit{game} (TG) is a tuple
$\game = (\automaton,\Sigma_1, \Sigma_2)$ where
$\automaton= (L, \ell_\init, C, \Sigma, I, E)$ is a TA and
$(\Sigma_1, \Sigma_2)$ is a partition of $\Sigma$.
We refer to actions in $\Sigma_i$ as $\player_i$ actions for $i\in\{1, 2\}$.

Recall a move is a pair $(d, a)\in\IR_{\geq 0}\times (\Sigma\cup\{\bot\})$.
Let $S$ denote the set of states of $\mathcal{T}(\automaton)$.
In each state $s= (\ell, \nu)\in S$, the moves available
to $\player_1$ are the elements of the set $M_1(s)$ where
\[M_1(s) = \big\{ (d,a) \in \mathbb{R}_{\geq 0} \times (\Sigma_1 \cup \{ \bot \}) \mid \exists s', s \xrightarrow{d,a} s'  \big\}\]
contains moves with $\player_1$ actions and delay moves that are
enabled in $s$. The set $M_2(s)$ is defined
analogously with $\player_2$ actions. We write $M_1$ and $M_2$ for
the set of all moves of $\player_1$ and $\player_2$ respectively.

At each state $s$ along a play, both players simultaneously
select a move $(d^{(1)}, a^{(1)})\in M_1(s)$ and $(d^{(2)}, a^{(2)})\in M_2(s)$.
Intuitively, the fastest player gets to act and in case of
a tie, the move is chosen non-deterministically. This is
formalized by the
\emph{joint destination function} $\delta: S\times M_1\times M_2\to 2^S$,
defined by
\[\delta(s, (d^{(1)}, a^{(1)}), (d^{(2)}, a^{(2)})) = \begin{cases}
    \{s' \in S\mid s \xrightarrow{d^{(1)}, a^{(1)}} s'\} & \text{if } d^{(1)} < d^{(2)} \\
    \{s' \in S\mid s \xrightarrow{d^{(2)}, a^{(2)}} s'\} & \text{if } d^{(1)} > d^{(2)} \\
    \{s' \in S\mid s \xrightarrow{d^{(i)}, a^{(i)}} s', i=1, 2\} & \text{if } d^{(1)} = d^{(2)}.
  \end{cases}\]
For $m^{(1)}=(d^{(1)}, a^{(1)})\in M_1$ and
$m^{(2)}=(d^{(2)}, a^{(2)})\in M_2$, we write
$\delay(m^{(1)}, m^{(2)})=\min\{d^{(1)}, d^{(2)}\}$ to denote the delay occurring when
$\player_1$ and $\player_2$ play $m^{(1)}$ and $m^{(2)}$ respectively.

A play is defined similarly to a path: it is a finite or infinite sequence
of the form
$s_0(m_0^{(1)}, m_0^{(2)})s_1(m_1^{(1)}, m_1^{(2)})\ldots\in
S((M_1\times M_2)S)^*\cup (S(M_1\times M_2))^\omega $ where
for all indices $j$, $m_j^{(i)}\in M_i(s_j)$ for $i\in\{1, 2\}$, and for $j>0$,
$s_{j}\in \delta(s_{j-1}, m_{j-1}^{(1)}, m_{j-1}^{(2)})$. A play is \textit{initial} if
$s_0=s_{\init}$. For a finite play
$\pi = s_0\ldots s_n$, we set $\last(\pi) = s_n$. For an infinite play
$\pi = s_0\ldots$, we write $\pi_{|n} = s_0(m_0^{(0)}, m_0^{(1)})\ldots s_n$.
A play follows a path in the TA, but there need not be a unique
path compatible with a play:  if along a play, at the $n$th step, the moves
of both players share the same delay and target state, either move can label
the $n$th transition in a matching path.

Similarly to paths, an infinite play
$\pi=(\ell_0, \nu_0)(m_0^{(1)}, m_0^{(2)})\cdots$
is \textit{time-divergent} if and only if $(\nu_j(\gamma))_{j\in\IN}$
is not bounded from above. Otherwise, we say a play is \textit{time-convergent}.
We define the following sets:
$\plays(\game)$ for the set of plays of $\game$;
$\plays_\mathit{fin}(\game)$ for the set of finite plays of $\game$;
$\plays_\infty(\game)$ for the set of time-divergent plays of $\game$.
We also write $\plays(\game, s)$ to denote plays starting in state $s$ of
$\mathcal{T}(\automaton)$.

Note that our games are built on \textit{deterministic} TAs. From a modeling standpoint, this is not restrictive, as we can simulate a non-deterministic TA through the actions of $\player_2$.
\subparagraph*{Strategies.} A strategy for $\player_i$ is a function describing
which move a player should use based on a play history.
Formally, a strategy for $\player_i$ is a function
$\sigma_i\colon \plays_\mathit{fin}(\game)\to M_i$ such that for all
$\pi\in\plays_\mathit{fin}(\game)$, $\sigma_i(\pi)\in M_i(\last(\pi))$. This last
condition requires that each move given by a strategy
be enabled in the last state of a play.

A play $s_0(m_0^{(1)}, m_0^{(2)})s_1\ldots$ is said to be consistent with
a $\player_i$-strategy $\sigma_i$ if for all indices $j$,
$m_j^{(i)} = \sigma_i(\pi_{|j})$. Given a $\player_i$-strategy $\sigma_i$,
we define $\outcome_i(\sigma_i)$ (resp.~$\outcome_i(\sigma_i, s)$)
to be the set of plays (resp.~set of plays starting in state $s$)
consistent with $\sigma_i$.

A $\player_i$-strategy $\sigma_i$ is \textit{move-independent} if the move it
suggests depends only on the sequence of states seen in the play.
Formally, $\sigma_i$ is move-independent if for all finite plays
$\pi = s_0(m_0^{(1)}, m_0^{(2)})s_1\ldots s_k$ and
$\tilde{\pi} =
\tilde{s}_0(\tilde{m}_0^{(1)}, \tilde{m}_0^{(2)})\tilde{s}_1\ldots \tilde{s}_k$
if $s_n=\tilde{s}_n$ for all $n\in\{0, \ldots, k\}$, then
$\sigma_i(\pi) = \sigma_i(\tilde{\pi})$. We use move-independent strategies
in the proof of our reduction to relabel some moves of a play without
affecting the suggestions of the strategy.

\subparagraph*{Objectives.} An objective represents the
property we desire on paths of a TA or a goal of a player in a TG.
Formally, we define an \textit{objective} as a set
$\Psi\subseteq \paths(\automaton)$ of infinite paths (when studying TAs)
or a set $\Psi\subseteq \plays(\game)$ of infinite plays (when studying TGs).
An objective is
\textit{state-based} (resp. \textit{location-based}) if it depends solely
on the sequence of states (resp. of locations) in a path or play. Any
location-based objective is state-based.

\begin{remark}\label{remark:paths_vs_plays}
  In the sequel, we present objectives exclusively as sets of plays.
  Definitions for paths are analogous as all the objectives defined hereafter
  are state-based.
\end{remark}

We use the following classical location-based objectives. A
\textit{reachability} objective for a set $F$ of locations is
the set of plays that pass through a
location in $F$. The complement of a reachability objective is a \textit{safety}
objective; given a set $F$, it is the set of plays that never
visit a location in $F$. A \textit{Büchi} objective for a set $F$ contains all
plays that pass through locations in $F$ infinitely often and the complement
\textit{co-Büchi} objective consists of plays traversing locations in $F$ finitely often.
The \textit{parity} objective for a priority function $p$ over the set of locations
requires that the smallest priority seen infinitely often is even.

Fix $F$ a set of locations and $p$ a priority function.
The aforementioned objectives are formally defined as follows.
\begin{itemize}
\item $\mathsf{Reach}(F) = \{
  (\ell_0, \nu_0) (m_0^{(1)}, m_0^{(2)})\ldots\in
  \plays(\game)\mid \exists\, n,\, \ell_n\in F\}$;
\item $\mathsf{Safe}(F) = \{
  (\ell_0, \nu_0) (m_0^{(1)}, m_0^{(2)})\ldots\in
  \plays(\game)\mid \forall\, n,\, \ell_n\notin F\}$.
\item $\mathsf{B\ddot{u}chi}(F) = \{
  (\ell_0, \nu_0) (m_0^{(1)}, m_0^{(2)})\ldots\in
  \plays(\game)\mid
  \forall\, j,\,\exists\, n\geq j,\, \ell_n\in F\}$;
\item $\mathsf{coB\ddot{u}chi}(F) = \{
  (\ell_0, \nu_0) (m_0^{(1)}, m_0^{(2)})\ldots\in
  \plays(\game)\mid \exists \, j, \,\forall\, n\geq j,\, \ell_n\notin F\}$;
\item $\mathsf{Parity}(p) = \{
  (\ell_0, \nu_0) (m_0^{(1)}, m_0^{(2)})\ldots\in
  \plays(\game)\mid (\liminf_{n\to\infty}p(\ell_n))\bmod 2 =0\}$.
\end{itemize}

\subparagraph*{Winning conditions.}
In games, we distinguish objectives and \emph{winning conditions}.
We adopt the definition of \cite{AlfaroFHMS03}. Let $\Psi$ be an objective.
It is desirable to have victory be achieved in a physically meaningful way:
for example, it is unrealistic to have a safety objective be achieved by
stopping time. This motivates a restriction to time-divergent plays.
However, this requires $\player_1$ to force the divergence of plays, which
is not reasonable, as $\player_2$ can stall using delays with zero time
units. Thus we also declare winning time-convergent plays where
$\player_1$ is \emph{blameless}.
Let $\textsf{Blameless}_1$ denote the set of $\player_1$-blameless
plays, which we define in the following way.

Let $\pi = s_0(m_0^{(1)}, m_0^{(2)})s_1\ldots$ be a (possibly finite) play. We say
$\player_1$ is \textit{not responsible} (or not to be blamed) for
the transition at step $n$ in $\pi$ if either $d_n^{(2)} < d_n^{(1)}$
($\player_2$ is faster) or $d_n^{(1)} = d_n^{(2)}$ and
$s_n\xrightarrow{d_n^{(1)}, a_n^{(1)}}s_{n+1}$
does not hold in $\mathcal{T}(\automaton)$
($\player_2$'s move was selected and did not have the same
target state as $\player_1$'s) where $m_n^{(i)} = (d_n^{(i)}, a_n^{(i)})$ for
$i\in\{1, 2\}$. 
The set $\textsf{Blameless}_1$ is formally defined as the set of infinite
plays $\pi$ such that there is some $j$ such that
for all $n\geq j$, $\player_1$ is not responsible for the transition at step
$n$ in $\pi$.

Given an objective
$\Psi$, we set the winning condition $\mathsf{WC}_1(\Psi)$ for $\player_1$ to be
the set of plays
\[\mathsf{WC}_1(\Psi) = (\Psi\cap\plays_\infty(\game))
  \cup (\textsf{Blameless}_1\setminus \plays_\infty(\game)).\]
Winning conditions for $\player_2$ are defined by exchanging the
roles of the players in the former definition.

We consider that the two players are adversaries and have opposite objectives,
$\Psi$ and $\neg \Psi$ (shorthand for $\plays(\game)\setminus \Psi$).
Let us note that there may be plays $\pi$ such that
$\pi\notin\mathsf{WC}_1(\Psi)$ and $\pi\notin\mathsf{WC}_2(\neg\Psi)$, e.g.,
any time-convergent play in which neither player is blameless.

A \emph{winning strategy} for $\player_i$ for an objective $\Psi$ from
a state $s_0$ is a strategy $\sigma_i$ such that
$\outcome_i(\sigma_i, s_0)\subseteq \mathsf{WC}_i(\Psi)$.
Move-independent strategies are known to suffice for timed automaton games
with state-based objectives \cite{AlfaroFHMS03}.

\subparagraph*{Decision problems.}
We consider two different problems for an objective $\Psi$.
The first is the \emph{verification problem} for $\Psi$, which asks
given a timed automaton whether all \emph{time-divergent initial} paths
satisfy the objective.
Second is the \emph{realizability problem}, which
asks whether in a timed automaton game with objective $\Psi$, $\player_1$
has a winning strategy from the initial state.

\section{Window objectives}\label{section:windows}
We consider the \emph{fixed window parity} and \emph{direct fixed window parity}
problems from~\cite{BruyereHR16} and adapt the discrete-time requirements from
their initial formulation to dense-time requirements for TAs and TGs.
Intuitively, a direct fixed  window parity objective for some bound $\lambda$
requires that at all points along a play or a path, we see a window of size
less than $\lambda$ in which the smallest priority is even. The (non-direct)
fixed window parity objective requires that the direct objective holds for some
suffix. In the sequel, we drop ``fixed'' from the name of these objectives.

In this section, we formalize the timed window parity objective in
TGs as sets of plays. The definition for paths of TAs is analogous
(see Remark~\ref{remark:paths_vs_plays}).
First, we define the \textit{timed good window objective},
which formalizes
the notion of good windows. Then we introduce the timed window parity
objective and its direct variant. We compare these objectives to the parity
objective and argue that satisfying a window objective implies satisfying a
parity objective, and that window objectives do not coincide with
parity objectives in general, via an example. We conclude this section by
proving some useful properties of this objective.

For this entire section, we fix a TG $\game = (\automaton, \Sigma_1, \Sigma_2)$
where $\automaton = (L, \ell_\init, C, \Sigma_1\cup\Sigma_2, I, E)$,
a priority function $p\colon L\to \{0, \ldots, d-1\}$ and
a bound $\lambda\in\IN\setminus\{0\}$ on the size of windows.

\subsection{Definitions}

\subparagraph{Good windows.} A window objective is based on a notion of
good windows. Intuitively, a good
window for the parity objective is a fragment of a play in which less than
$\lambda$ time units pass and the smallest priority of the locations appearing
in this fragment is even.

The timed good window objective encompasses plays in which there is a good
window at the start of the play.
We formally define the \textit{timed good window (parity) objective} as the set
\begin{align*}
    \mathsf{TGW}(p, \lambda)= \big\{(\ell_0, \nu_0)(m^{(1)}_0,m^{(2)}_0) \ldots
    \in\plays(\game)\mid 
    \exists \, j\in\IN,\, 
    &\min_{0\leq k\leq j}p(\ell_k)\bmod 2 = 0 \\ &
    \land  \nu_j(\gamma)-\nu_0(\gamma)<\lambda\big\}.
\end{align*}
The timed good window objective is a state-based objective.

We introduce some terminology related to windows.
Let
$\pi=(\ell_0, \nu_0)(m_0^{(1)}, m_0^{(2)})(\ell_1, \nu_1)\ldots$
be an infinite play. We say that
the window opened at step $n$ \textit{closes} at step $j$ if
$\min_{n\leq k\leq j}p(\ell_k)$ is even and for all $n\leq j'<j$,
$\min_{n\leq k\leq j'}p(\ell_k)$ is odd. 
Note that, in this case, we must have $\min_{n\leq k\leq j}p(\ell_k)=p(\ell_j)$.
In other words, a window closes
when an even priority smaller than all other priorities in the window
is encountered.
The window opened at step $n$ is said to \textit{close immediately} if
$p(\ell_n)$ is even.

If a window does not close within $\lambda$ time units, we refer to it as
a \emph{bad} window: the window opened at step $n$ is a bad window if
there is some $j^\star\geq n$ such that $\nu_{j^\star}(\gamma)-\nu_n(\gamma)\geq\lambda$
and for all $j\geq n$, if $\nu_j(\gamma)-\nu_n(\gamma)<\lambda$, then
$\min_{n\leq k\leq j}p(\ell_k)$ is odd.

\subparagraph{Direct timed window objective.} 
The direct window parity objective in graph games requires that every suffix
of the play belongs to the good window objective.
To adapt this objective to a dense-time setting, we must express that
\textit{at all times, we have a good window}. We require that this property
holds not only at states which appear explicitly along plays, but also
in the continuum between them (during the delay within a location). To this
end, let us introduce a notation for suffixes of play.

Let $\pi = (\ell_0, \nu_0)(m_0^{(1)}, m_0^{(2)})(\ell_1, \nu_1)\ldots
\in \plays(\game)$ be a play. For all $i\in\{1, 2\}$ and all $n\in\IN$,
write $m_n^{(i)} = (d_n^{(i)}, a_n^{(i)})$ and
$d_n =  \delay(m_n^{(1)}, m_n^{(2)}) = \nu_{n+1}(\gamma) -\nu_n(\gamma)$.
For any $n\in\IN$ and $d\in [0, d_n]$,
let $\pi_{n\to}^{+d}$ be the \textit{delayed suffix} of $\pi$
starting in position $n$ delayed by $d$ time units, defined as\[\pi_{n\to}^{+d} = (\ell_n, \nu_n + d)((d_n^{(1)}-d, a_n^{(1)}), (d_n^{(2)}-d, a_n^{(2)}))
  (\ell_{n+1}, \nu_{n+1}) (m_{n+1}^{(1)}, m_{n+1}^{(2)})\ldots\]
If $d=0$, we write $\pi_{n\to}$ rather than $\pi^{+0}_{n\to}$.

Using the notations above, we define the \emph{direct timed window (parity)
  objective} as the set
\[\mathsf{DTW}(p, \lambda) = \{\pi\in \plays(\game)\mid\forall\, n\in\IN,\,
  \forall\, d\in [0, d_n],\, \pi_{n\to}^{+d}\in\mathsf{TGW}(p, \lambda)\}.\]

The direct timed window  objective is state-based: the timed good window
objective is state-based and the delays $d_n$ are encoded in states (measured
by clock $\gamma$), thus all conditions in the definition of the direct timed
window  objective depend only the sequence of states of a play.

A good window for a delayed suffix $\pi_{n\to}^{+d}$ can be expressed using
exclusively indices
from the play $\pi$. In fact, $\pi_{n\to}^{+d}\in \mathsf{TGW}(p, \lambda)$
if and only if there is some $j\geq n$ such that $\min_{n\leq k\leq j}p(\ell_k)$
is even and $\nu_j(\gamma) - \nu_n(\gamma) - d < \lambda$.
We use this characterization
to avoid mixing indices from plays $\pi$ and $\pi_{n\to}^{+d}$ in proofs.

\subparagraph{Timed window objective.}
We define the \emph{timed window (parity) objective} as a prefix-independent
variant of the direct timed window  objective. Formally, we let
\[\mathsf{TW}(p, \lambda)= \{\pi\in \plays(\game)\mid \exists\, n\in\IN,\,
  \pi_{n\to}\in \mathsf{DTW}(p, \lambda)\}.\]
The timed window objective requires the direct timed window objective to hold
from some point on. This implies that the timed window objective is state-based.

\subsection{Comparison with parity objectives}
Both the direct and non-direct timed window  objectives reinforce the
parity objective with time bounds. It can easily be shown that
satisfying the direct timed window  objective implies satisfying
a parity objective. Any odd priority seen along a play in
$\mathsf{DTW}(p, \lambda)$ is answered within $\lambda$ time units by a
smaller even priority. Therefore, should any odd priority appear infinitely
often, it is followed by a smaller even priority. As the set of priorities
is finite, there must be some smaller even priority appearing infinitely
often. This in turn implies that the parity objective is fulfilled.
Using prefix-independence of the parity objective, we can also conclude that
satisfying the non-direct timed window  objective implies
satisfaction of the parity objective.

However, in some cases, the timed window  objectives may not hold
even though the parity objective holds. For simplicity, we provide
an example on a TA, rather than a TG. Consider
the timed automaton $\mathcal{A}$  depicted in Figure~\ref{fig:paritynotwindow}.

\begin{figure}
\centering
  \captionof{figure}{Timed automaton $\automaton$. Edges are labeled
    with triples guard-action-resets. Priorities are beneath
    locations. The initial state is denoted by an incoming arrow with no origin.} \label{fig:paritynotwindow}
  \begin{tikzpicture}[shorten <= 1pt, node distance=4cm, initial text=,
    scale=0.8, every node/.style={transform shape},
    every state/.style={minimum size=1.5cm}]
    \node[state, initial, align=center] (l0) {$\ell_0$ \\ $x\leq 2$};
    \node[align=center, below of=l0, node distance=1.1cm] {$1$};
    \node[state, align=center, right of=l0] (l1) {$\ell_1$ \\ $\true$};
    \node[align=center, below of=l1, node distance=1.1cm] {$2$};
    \node[state, align=center, right of=l1] (l2) {$\ell_2$ \\ $x\leq 2$};
    \node[align=center, below of=l2, node distance=1.1cm] {$0$};
    \path[->] (l0) edge node[align=center, above] {$(\true, a, \varnothing)$} (l1);
    \path[->] (l1) edge node[align=center, above] {$(\true, a, \{x\})$} (l2);
    \path[->] (l2) edge[bend right] node[align=center, above] {$(\true, a, \{x\})$} (l0);
  \end{tikzpicture}
\end{figure}
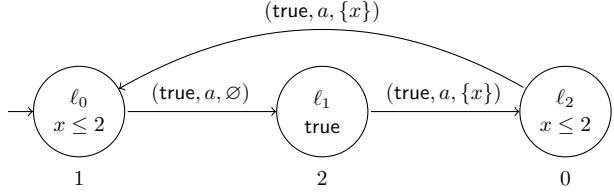

All time-divergent paths of $\automaton$ satisfy the parity objective. We can
classify time-divergent paths in two families: either $\ell_2$ is visited
infinitely often, or from some point on only delay transitions are taken in
$\ell_1$. In the former case, the smallest priority seen infinitely often is
$0$ and in the latter case, it is $2$.
  
However, there is a path $\pi$ such that for all window sizes
$\lambda\in\IN\setminus\{0\}$, $\pi$ violates the direct and non-direct timed window
 objectives.
Initialize $n$ to $1$. This path can be described by the following loop:
play action $a$ in $\ell_0$ with delay $0$, followed by action $a$ with delay
$n$ in $\ell_1$ and action $a$ in $\ell_2$ with delay $0$, increase $n$ by $1$ and repeat.
The window opened in $\ell_0$ only closes when location $\ell_2$ is entered.
At the $n$-th step of the loop, this window closes after $n$ time units. As
we let $n$ increase to infinity, there is no window size $\lambda$ such that
this path satisfies the direct and non-direct timed window  objectives for $\lambda$.

This example demonstrates the interest of reinforcing parity objectives with
time bounds; we can enforce that there is a bounded delay between an
odd priority and a smaller even priority in a path.
\subsection{Properties of window objectives}
We present several properties of the timed window  objective.
First, we show that we need only check good windows for non-delayed
suffixes $\pi_{n\to}$.
Once this property is proven, we move on to the inductive property
of windows, which is the crux of the reduction in the next section.
This inductive property states that when we close a window in less than $\lambda$ time units
all other windows opened in the meantime also close in less than
$\lambda$ time units.

The definition of the direct timed window  objective requires
checking uncountably many windows. This can be reduced to a countable
number of windows: those opened when entering states appearing along a play.
Let us explain why no information is lost through such a restriction. We rely on
a timeline-like visual representation given in Figure~\ref{fig:countable}.
Consider a window that does not close immediately and is opened in some state
of the play delayed by $d$ time units, of the form $(\ell_n, \nu_n + d)$
(depicted by the circle at the start of the bottom line of Figure~\ref{fig:countable}). This implies that the priority of $\ell_n$ is odd, otherwise this window
would close immediately.
Assume the window opened at step $n$ closes at step $j$
(illustrated by the middle line of the figure) in less than $\lambda$ time
units. As the priority of $\ell_n$ is odd, we must have $j\geq n+1$ (i.e., the
window opened at step $n$ is still open as long as $\ell_n$ is not left).
These lines cover the same locations, i.e., the set of locations appearing along
the time-frame given by both the dotted and dashed lines coincide.
Thus, the window opened $d$ time units after step $n$ closes in at most $\lambda-d$
time units, at the same time as the window opened at step $n$.

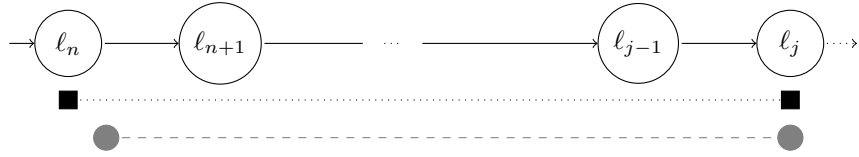
\begin{figure}[h]
\centering
\captionof{figure}{A timeline representation of a play. Circles with labels indicate entry in a location. The dotted line underneath represents a window opened at step $n$ and
  closed at step $j$ and the dashed line underneath the window opened $d$ time units after step $n$.} \label{fig:countable}
  \begin{tikzpicture}[shorten <= 1pt, node distance=0.5cm, scale=0.8, initial text=]
    \node [state, initial] (first) {$\ell_n$};
    \node [state, right of=first, node distance=2cm] (node2) {$\ell_{n+1}$};
    \node [right of=node2, node distance=2cm] (phantom1) {};
    \node [right of=phantom1, node distance=0.5cm] (phantom2) {};
    \node [state, right of=phantom2, node distance=3cm] (node3) {$\ell_{j-1}$};
    \node [state, right of=node3, node distance=2cm] (node4) {$\ell_{j}$};
    \node [right of=node4, node distance=1cm] (phantom3) {};
    \node [draw, fill=black, below of=first, node distance=0.75cm] (window1) {};
    \node [draw, fill=black, below of=node4, node distance=0.75cm] (window2) {};
    \node [below of=window1] (phantomwindow) {};
    \node [draw=gray, shape=circle, fill=gray, right of=phantomwindow] (windowr21) {};
    \node [draw=gray, shape=circle, fill=gray, below of=window2] (windowr22) {};
    \draw [->] (first) -- (node2);
    \draw (node2) -- (phantom1);
    \draw [dotted] (phantom1)-- (phantom2);
    \draw [->] (phantom2) -- (node3);
    \draw [->] (node3) -- (node4);
    \draw [->, dotted] (node4) -- (phantom3);
    \draw [dotted] (window1) --(window2);
    \draw [dashed, gray] (windowr21) -- (windowr22);
  \end{tikzpicture}
\end{figure}

\begin{lemma}\label{lemma:discrete_twpo}
  Let $\pi=(\ell_0, \nu_0)(m_0^{(1)}, m_0^{(2)})\ldots\in\plays(\game)$ and
  $n\in\IN$. Let $d_n$ denote $\delay(m_n^{(1)}, m_n^{(2)})$.
  Then $\pi_{n\to}\in\mathsf{TGW}(p, \lambda)$ if and only if for all
  $d\in[0, d_n]$, $\pi_{n\to}^{+d}\in \mathsf{TGW}(p, \lambda)$.
  Furthermore,
  $\pi\in\mathsf{DTW}(p, \lambda)$ if
  and only if for all $n\in\IN$, $\pi_{n\to}\in \mathsf{TGW}(p, \lambda)$.
\end{lemma}
\begin{proof}
  Assume for all $d\in [0, d_n]$,
  $\pi_{n\to}^{+d}\in \mathsf{TGW}(p, \lambda)$ holds.
  Selecting $d=0$ yields $\pi_{n\to}\in\mathsf{TGW}(p, \lambda)$.

  Conversely, assume that $\pi_{n\to}\in\mathsf{TGW}(p, \lambda)$.
  Let $d\in [0, d_n]$. By definition of timed good
  window objectives, there is some $j\geq n$ such that
  $\nu_j(\gamma) - \nu_n(\gamma)< \lambda$ and $\min_{n\leq k\leq j}p(\ell_k)$ is
  even. The fact that $\pi_{n\to}^{+d}\in \mathsf{TGW}(p, \lambda)$ follows
  immediately from the chain of inequalities
  $\nu_j(\gamma) - \nu_n(\gamma) - d \leq \nu_j(\gamma) - \nu_n(\gamma)<\lambda$.

  The last claim of the lemma follows immediately from the first part of the
  lemma and the definition of direct timed window  objectives.
\end{proof}

In turn-based games on graphs, window objectives exhibit an inductive
property:
when a window closes, all subsequently opened windows close (or were closed
earlier)~\cite{BruyereHR16}.
This is also the case for the timed variant. A window closes
when an even priority smaller than all priorities seen in the window
is encountered.
This priority is also smaller than priorities in all windows opened
in the meantime, therefore they must close at this point (if they are not
yet closed). We state this property only for windows opened at steps along
the run and neglect the continuum in between due to Lemma~\ref{lemma:discrete_twpo}.

\begin{lemma}[Inductive property]\label{lemma:inductive_prop}
  Let $\pi = (\ell_0, \nu_0)(m_0^{(1)}, m_0^{(2)})(\ell_1, \nu_1)\ldots\in
  \plays(\game)$.
  Let $n\in \IN$. Assume the window opened at step $n$ closes at step $j$
  and $\nu_j(\gamma)-\nu_n(\gamma) < \lambda$.
  Then, for all $n\leq i\leq j$, $\pi_{i\to}\in\mathsf{TGW}(p, \lambda)$.
\end{lemma}
\begin{proof}
  Fix $i\in\{n, \ldots, j\}$.
  The sequence $(\nu_k(\gamma))_{k\in\IN}$ is non-decreasing, which implies that
  $\nu_{j}(\gamma) - \nu_i(\gamma) \leq \nu_{j}(\gamma)-\nu_n(\gamma)< \lambda$.
  It remains to show that $\min_{i\leq k\leq j}p(\ell_k)$ is even. As the
  window opened at step $n$ closes at step $j$, we have
  $\min_{n\leq k\leq j}p(\ell_k) = p(\ell_j)$ and $p(\ell_j)$ is even.
  We have the inequalities $p(\ell_{j}) \leq
    \min_{n\leq k\leq j}p(\ell_k)\leq \min_{i\leq k\leq j}p(\ell_k)
    \leq p(\ell_{j});$
  the first follows from above, the second because we take a minimum
  of a smaller set and the third by definition of minimum.
  Thus $\min_{i\leq k\leq j}p(\ell_k) = p(\ell_{j})$ is even, ending the proof.
\end{proof}

It follows from this inductive property that it suffices to keep
track of one window at a time when checking whether a play
satisfies the (direct) timed window  objective.

\section{Reduction}\label{section:reduction}
We establish in this section that the realizability
(resp.~verification) problem for the direct/non-direct timed window parity
objective can be reduced to the the realizability
(resp.~verification) problem for safety/co-Büchi objectives on an expanded
TG (resp.~TA). Our reduction uses the same construction of
an expanded TA for both the verification and realizability problems.
A state of the expanded TA describes the status of a window,
allowing the detection of bad windows. This section is divided in three
parts.

Firstly, we describe how a TA can be expanded with window-related
information.
Then we show that time-divergent plays in a TG and its expansion can
be related, by constructing two (non-bijective) mappings, in a manner such
that a time-divergent play in the base TG satisfies the direct/non-direct
timed window parity objective if and only if its related play in the expanded
TG satisfies the safety/co-Büchi objective.
These results developed for plays are (indirectly) applied to paths in
order to show the correctness of the reduction for the verification problem.
Thirdly, we establish that the mappings developed in the second part can be
leveraged to translate strategies in TGs, and prove that the presented
translations preserve winning strategies, proving correctness of the reduction
for TGs.

For this section, we fix a TG $\game=(\automaton, \Sigma_1, \Sigma_2)$
with TA $\automaton = (L, \ell_\init, C, \Sigma, I, E)$, a priority function
$p$ and a bound $\lambda$ on the size of windows.

\subsection{Encoding the objective in an automaton}\label{section:reduction:construction}

To solve the verification and realizability problems for the
direct/non-direct timed window  objective, we rely on a reduction to
a safety/co-Büchi objective in an expanded TA. The inductive property
(Lemma~\ref{lemma:inductive_prop}) implies that it suffices to
keep track of one window at a time when checking a window objective.
Following this, we encode the status of a window in the TA.

A window can be summarized by two characteristics: the lowest
priority within it and for how long it has been open. To keep track of the first
trait, we encode the lowest priority seen in the current window
in locations of the TA.
An expanded location is a pair $(\ell, q)$ where $q\in\{0, \ldots, d - 1\}$;
the number $q$ represents the smallest priority in the window currently under
consideration. We say a pair $(\ell, q)$ is an even (resp.~odd) location if
$q$ is even (resp.~odd).
To measure how long a window is opened, we use an additional clock $z\notin C$
that does not appear in $\automaton$. This clock is reset whenever a new window
opens or a bad window is detected.

The focus of the reduction is over time-divergent plays. Some
time-convergent
plays may violate a timed good window objective without ever seeing a bad
window, e.g., when time does not progress up to the supplied window size.
Along time-divergent plays however, the lack of a good window at any point equates
to the presence of a bad window. We encode the (resp.~direct) timed window
objective as a co-Büchi (resp.~safety) objective. Locations
to avoid in both cases indicate bad windows and are additional
expanded locations $(\ell, \bad)$, referred to as \textit{bad locations}.
We introduce two new actions $\beta_1$ and $\beta_2$, one per player, for
entering and exiting bad locations. While only the action $\beta_1$ is
sufficient for the reduction to be correct, introducing two actions allows
for a simpler correctness proof in the case of TGs; we can exploit the fact that
$\player_2$ can enter and exit bad locations. We use two new actions no matter
the considered problem: this enables us to use the same expanded TA construction for both the verification problem and realizability problem.

It remains to discuss how the initial location, edges and invariants of
an expanded TA are defined. We discuss edges and invariants for each type
of expanded location, starting with even locations, then odd locations and
finally bad locations. Each rule we introduce hereafter is followed
by an application on an example. We depict the TA of
Figure~\ref{fig:paritynotwindow} and the reachable fragment of its expansion
in Figure~\ref{fig:paritynotwindow:expanded} and use these TAs for our example.
For this explanation, we use the terminology of TAs (paths) rather than
that of TGs (plays).

The initial location of an expanded TA encodes the window opened at the
start of an initial path of the original TA. This window contains only a single
priority, that is the priority of the initial location of the original TA.
Thus, the initial location of the expanded TA is the expanded location
$(\ell_\init, p(\ell_\init))$. In the case of our
example, the initial location is  $\ell_0$ and the priority of $\ell_0$ is $1$,
thus the initial location of the expanded TA is $(\ell_0, 1)$.

Even expanded locations encode windows that are closed and do not need to
be monitored anymore.
Therefore, the invariant of an even expanded location is unchanged from the
invariant of the original location in the original TA.
Similarly, we do not add any additional constraints on the edges leaving even expanded
locations.
Leaving an even expanded location means opening a new window:
any edge leaving an even expanded location has an expanded location
of the form $(\ell, p(\ell))$ as its target
($p(\ell)$ is the only priority occurring in the
new window) and resets $z$ to start measuring
the size of the new window.
For example, in Figure~\ref{fig:paritynotwindow:expanded}, the edge
from $(\ell_2, 0)$ to $(\ell_0, 1)$ of the expanded TA
is obtained this way from the edge from $\ell_2$ to $\ell_0$ in the original
TA.

Odd expanded locations represent windows that are still open.
The clock $z$ measures how long a window has been opened. If $z$ reaches
$\lambda$ in an odd expanded location, that equates to a bad window in the
original TA. In this case, we force time-divergent paths of the expanded TA
to visit a bad location.
This is done in three steps. We strengthen the invariant of odd expanded
locations to prevent $z$ from exceeding $\lambda$. We also
disable the edges that leave odd expanded locations and do not go to a bad
location whenever $z=\lambda$ holds, by reinforcing the guards of such edges
by $z<\lambda$. Finally, we include two edges to a bad location
(one per additional action $\beta_1$ and $\beta_2$), which can only be
used whenever there is a bad window, i.e., when $z=\lambda$.
In the case of our example, if $z$ reaches $\lambda$ in $(\ell_0, 1)$,
we redirect the path to location $(\ell_0, \bad)$, indicating a window has
not closed in time in $\ell_0$. When $z$ reaches $\lambda$ in $(\ell_0, 1)$,
no more non-zero delays are possible, the edge from $(\ell_0, 1)$ to
$(\ell_1, 1)$ is disabled and only the edges to $(\ell_0, \bad)$ are enabled.

When leaving an odd expanded location using an edge, assuming we do not
go to a bad location, the smallest priority of the window has to be updated.
The new smallest priority is the minimum between the smallest
priority of the window prior to traversing the edge and the priority of
the target location. In our example for instance,
the edge from $(\ell_1, 1$) to $(\ell_2, 0)$ is derived from the edge from
$\ell_1$ to $\ell_2$ in the original TA.
As the priority of $\ell_2$
is 0 and is smaller than the current smallest priority of the window encoded
by location $(\ell_1, 1)$, the smallest priority of the window is updated
to $0=\min\{1, p(\ell_2)\}$ when traversing the edge. Note that we do not
reset $z$ despite the encoded window closing upon entering $(\ell_2, 0)$:
the value of $z$ does not matter while in even locations, thus there is no
need for a reset when closing the window.

A bad location $(\ell, \bad)$ is entered whenever a bad window is
detected while in location $\ell$. Bad
locations are equipped with the invariant $z=0$ preventing the passage of
time. In this way, for time-divergent paths, a new window is opened immediately
after a bad window is detected. For each additional action $\beta_1$ and
$\beta_2$, we add an edge exiting the bad location.
Edges leaving a bad location $(\ell, \bad)$ have as their target the expanded location
$(\ell, p(\ell))$; we reopen a window in the
location in which a bad window was detected. The clock $z$ is not reset by
these edges, as it was reset prior to entering the bad
location and the invariant $z=0$ prevents any non-zero delay in the bad
location. For instance, the edges
from $(\ell_1, \bad)$ to $(\ell_1, 2)$ in our example represent that
when reopening while in location $\ell_1$, the smallest priority of this window
is $p(\ell_1)=2$.

\begin{figure}
  \centering
  \caption{The TA of Figure~\ref{fig:paritynotwindow} (left) and the
    reachable fragment of its expansion (right).
    We write $\beta$ for  actions $\beta_1$ and $\beta_2$.}
  \label{fig:paritynotwindow:expanded}
    \begin{tikzpicture}[shorten <= 1pt, node distance=3cm, initial text=,
    scale=0.65, every node/.style={transform shape},
    every state/.style={minimum size=1.5cm}]
    \node[state, initial above, align=center] (l0) {$\ell_0$ \\ $x\leq 2$};
    \node[align=center, left of=l0, node distance=1.1cm] {$1$};
    \node[state, align=center, below of=l0] (l1) {$\ell_1$ \\ $\true$};
    \node[align=center, left of=l1, node distance=1.1cm] {$2$};
    \node[state, align=center, below of=l1] (l2) {$\ell_2$ \\ $x\leq 2$};
    \node[align=center, left of=l2, node distance=1.1cm] {$0$};
    \path[->] (l0) edge node[align=center, left] {$(\true, a, \varnothing)$} (l1);
    \path[->] (l1) edge node[align=center, left] {$(\true, a, \{x\})$} (l2);
    \path[->] (l2) edge[bend right] node[align=center, right] {$(\true, a, \{x\})$} (l0);
  \end{tikzpicture}
  \begin{tikzpicture}[shorten <= 1pt, node distance=5cm, initial text=,
    scale=0.65, every node/.style={transform shape},
    every state/.style={minimum size=2.5cm}]
    \node[state, initial, align=center] (l0) {$(\ell_0,1)$ \\ $x\leq 2\land z\leq\lambda$};
    \node[state, align=center, right of=l0] (l11) {$(\ell_1, 1)$ \\ $z\leq\lambda$};
    \node[state, align=center, below of=l0] (l0bad) {$(\ell_0, \bad)$ \\ $z=0$};
    \node[state, align=center, below of=l11] (l1bad) {$(\ell_1, \bad)$ \\ $z=0$};
    \node[state, align=center, right of=l11] (l20) {$(\ell_2, 0)$ \\ $x\leq 2$};
    \node[state, align=center, right of=l1bad] (l12) {$(\ell_1, 2)$ \\ $\true$};
    \path[->] (l0) edge node[align=center, above] {$(z< \lambda, a, \varnothing)$} (l11);
    \path[->] (l0) edge[bend left] node[align=center, right] {$(z = \lambda, \beta, \{z\})$} (l0bad);
    \path[->] (l0bad) edge[bend left] node[align=center, left] {$(\true,\beta, \varnothing)$} (l0);
    \path[->] (l11) edge  node[align=center, right] {$(z = \lambda, \beta, \{z\})$} (l1bad);
    \path[->] (l11) edge node[align=center, above] {$(z < \lambda, a, \{x\})$} (l20);
    \path[->] (l12) edge  node[align=center, right] {$(\true, a, \{x, z\})$} (l20);
    \path[->] (l1bad) edge node[align=center, above] {$(\true, \beta, \varnothing)$} (l12);
    \path[->] (l20) edge[bend right] node[align=center, above] {$(\true, a, \{x, z\})$} (l0);
  \end{tikzpicture}
\end{figure}
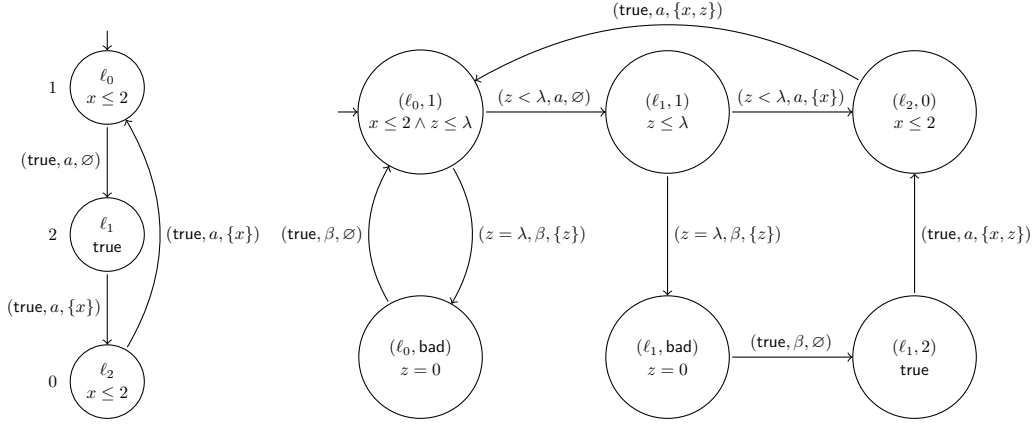

The expansion depends on the priority function $p$ and the
bound on the size of windows $\lambda$. Therefore, we write
$\automaton(p, \lambda)$ for the expansion.
The formal definition of $\automaton(p, \lambda)$ follows.
\begin{definition}\label{definition:singledim:extended_automaton}
  Given a TA $\automaton = (L, \ell_\init, C, \Sigma, I, E)$,
  the TA $\automaton(p, \lambda)$ is defined to be the
  TA $(L', \ell'_\init, C', \Sigma', I', E')$
  such that
  \begin{itemize}
  \item $L' = L\times (\{0, \ldots, d-1\}\cup \{\bad\})$;
  \item $\ell'_\init = (\ell_\init, p(\ell_\init))$;
  \item $C' = C\cup \{z\}$ where $z\notin C$ is a new clock;
  \item $\Sigma' = \Sigma\cup \{\beta_1, \beta_2\}$ is an expanded set
    of actions with special actions $\beta_1$, $\beta_2\notin\Sigma$
    for bad locations;
  \item $I'(\ell, q) = I(\ell)$ for all $\ell\in L$ and
    even $q\in\{0, \ldots, d-1\}$, $I'(\ell, q) = (I(\ell)\land z\leq\lambda)$
    for all $\ell\in L$ and odd $q\in\{0, \ldots, d-1\}$,
    and $I'(\ell, \bad) = (z=0)$
    for all $\ell\in L$;
  \item  the set of edges $E'$ of $\automaton(p, \lambda)$ is the smallest
    set satisfying the following rules:
    \begin{itemize}
    \item if $q$ is an even number and $(\ell, g, a, D, \ell')\in E$, then
      \[((\ell, q), g, a, D\cup\{z\}, (\ell', p(\ell'))\in E';\]
    \item if $q$ is an odd number and $(\ell, g, a, D, \ell')\in E$, then
      \[((\ell, q), (g\land z< \lambda), a,
        D, (\ell', \min\{q, p(\ell')\}))\in E';\]
    \item for all locations $\ell\in L$, odd $q$ and $\beta\in\{\beta_1, \beta_2\}$,
      \[((\ell, q), (z= \lambda), \beta, \{z\}, (\ell, \bad))\in E'\text{ and }
      ((\ell, \bad), \true, \beta, \varnothing, (\ell, p(\ell))\in E'.\]
    \end{itemize}
  \end{itemize}

  For a TG $\game = (\automaton, \Sigma_1, \Sigma_2)$,
  we set $\game(p, \lambda) = (\automaton(p, \lambda),
  \Sigma_1\cup\{\beta_1\}, \Sigma_2\cup\{\beta_2\})$.
\end{definition}
We write $(\ell, q, \bar{\nu})$ for states of
$\mathcal{T}(\automaton(p, \lambda))$ instead of $((\ell, q), \bar{\nu})$, for
conciseness. The bar over the valuation is a visual indicator of the different
domain. 

We say that a play $(\ell_0, q_0, \bar{\nu}_0)(m_0^{(1)}, m_0^{(2)})\cdots$ of
$\game(p, \lambda)$ is
\emph{well-initialized} if $q_0 = p(\ell_0)$ and $\bar{\nu}_0(z)=0$.
A well-initialized play can be seen as a play with a window opening at its
start. Any initial play in $\game(p, \lambda)$ is well-initialized.
Proving statements related to the direct timed window  objective
only for initial plays (rather than well-initialized plays) is too restrictive
to effectively apply them to
to the timed window  objective, as this objective
deals with suffixes. We later define expansions of plays as specific
well-initialized plays.
We write
$\badset=L\times\{\bad\}$ for the set of bad locations.

\subsection{Expanding and projecting plays}\label{section:reduction:mappings}
We prove that any play of $\game$ has an expansion in
$\game(p,\lambda)$, and conversely, any play in $\game(p,\lambda)$
projects to a play in $\game$. This is done by constructing an expansion
mapping and a projection mapping, both of which are shown to
behave well w.r.t.~our objectives (Lemma~\ref{theorem:window_mapping}).

\begin{remark}\label{remark:no_bijection}
  Note that we do not construct a bijection between the set of plays of
  $\game$ and the set of plays of $\game(p, \lambda)$. This cannot be
  achieved naturally due to the additional information encoded in
  the expanded automaton, and notably the presence of bad locations.
  We illustrate this by showing there are some plays of $\game(p, \lambda)$
  that are intuitively indistinguishable if seen as plays of $\game$.
  
  Consider the initial
  location $\ell_\init$ of $\game$, and assume that its priority is odd and its
  invariant is $\true$. Consider the
  initial play $\bar{\pi}_1$ of $\game(p, \lambda)$ where
  the actions $\beta_i$ are used by both players with a delay of $\lambda$
  at the start of the play and then only delay moves are taken in the
  reached bad location, i.e.,
  $\bar{\pi}_1=(\ell, p(\ell), \mathbf{0}^{C\cup \{z\}})
  ((\lambda, \beta_1), (\lambda, \beta_2))\big((\ell, \bad, \bar{\nu})
  ((0, \bot), (0, \bot))\big)^\omega$, where $\bar{\nu}(x)=\lambda$ for
  all $x\in C$ and
  $\bar{\nu}(z)=0$. As the actions $\beta_i$ and $z$ do not exist in
  $\game$, $\bar{\pi}_1$ cannot be discerned from the
  similar play $\bar{\pi}_2$
  of $\game(p, \lambda)$ where instead of using the actions $\beta_i$, delay
  moves were used instead, i.e.,
  $\bar{\pi}_2=(\ell, p(\ell), \mathbf{0}^{C\cup \{z\}})
  ((\lambda, \bot), (\lambda, \bot))((\ell, p(\ell), \bar{\nu}')
  ((0, \bot), (0, \bot)))^\omega$ with $\bar{\nu}'(x)=\lambda$ for all
  $x\in C\cup\{z\}$.

  This motivates using two mappings instead of a bijection to prove
  the correctness of our reduction.
\end{remark}

\subsubsection*{Expansion mapping} The expansion mapping
$\mathsf{Ex}\colon \plays(\game) \to
\plays(\game(p, \lambda))$ between plays of
$\game$ and of $\game(p, \lambda)$ is defined by an inductive
construction.
We construct expansions step by step. The rough idea
is to use the same moves in the play of $\game$ being expanded
and in its expansion in $\game(p, \lambda)$, as long as these moves do not
allow $z$ to
reach $\lambda$ in an odd location of $\game(p, \lambda)$, i.e., as long as
these moves do not make us see a bad window.
In fact the construction addresses how to proceed if a move enabled in $\game$
would allow $z$ to reach $\lambda$ in an odd location. If only one of the
two players, say $\player_i$, suggests a move with a large enough delay for
clock $z$ to reach
$\lambda$, then their adversary $\player_{3-i}$ preempts them and it suffices
to replace
$\player_i$'s move by any valid move with a larger delay than $\player_{3-i}$'s.
However,
if both players suggest moves with too large a delay, the expanded play
goes through a bad location (possibly multiple times) until enough time passes
and one of the two players can use their move (with the remaining delay)
in the expanded game.

Before presenting a formal construction of the expansion mapping, let us
describe the structure of the inductive step of the construction.
We number the different cases in the same order as they appear in the upcoming
formal definition.
Let $\pi = s_0(m_0^{(1)}, m_0^{(2)})\ldots s_{n+1}$ be a play of
$\game$, and assume the expansion of its prefix
$\pi_{|n} = s_0(m_0^{(1)}, m_0^{(2)})\ldots s_{n}$ has already been constructed.
We assume inductively that the last states of $\pi_{|n}$ and its expansion
$\ex(\pi_{|n})$
share the same location of the original TG and that their clock valuations
agree over $C$. Furthermore, we also inductively assume the last state of
$\ex(\pi_{|n})$ is not in a bad location.
Write $\bar{s} = \last(\ex(\pi_{|n})) = (\ell, q, \bar{\nu})$ and
$s = \last(\pi_{|n}) = (\ell, \nu)$. Denote by $d = \delay(m_n^{(1)}, m_n^{(2)})$
the delay of the last pair of moves of the players in $\pi$.

\begin{enumerate}
\item If $q$ is even, the same moves are available in $s$ and $\bar{s}$; the
  expansion can be extended using the pair of moves $(m_n^{(1)}, m_n^{(2)})$.
  \label{itemize:expansion:sketch:1}
  
\item If $q$ is odd, the invariant of the expanded location $(\ell, q)$ prevents
  $z$ from exceeding $\lambda$. We distinguish cases depending on whether the
  delay $d$ allows $z$ to reach $\lambda$.
  \begin{enumerate}
  \item If $\bar{\nu}(z) + d < \lambda$, then one of the players has offered
    a move enabled in $\bar{s}$: this move determines how to extend the play.
  \item Otherwise, $\bar{\nu}(z) + d \geq \lambda$. The construction makes
    the expansion go through location $(\ell, \bad)$. When $(\ell, \bad)$
    is exited, the path goes to $(\ell, p(\ell))$, the invariant of which
    depends on the parity of $p(\ell)$. We treat each case differently.
    \begin{enumerate}
    \item If $p(\ell)$ is even, then the invariant of $(\ell, p(\ell))$ matches
      that of $\ell$. Once $(\ell, \bad)$ is left, we reason similarly
      to case \ref{itemize:expansion:sketch:1}, using the moves with
      whatever delay remains.
    \item If $p(\ell)$ is odd, it may be required to go to a bad location more
      than once if the remaining delay after the first visit to the bad
      location exceeds $\lambda$. Once the remaining delay is strictly less
      than $\lambda$, we can operate as in case 2.a.
    \end{enumerate}
  \end{enumerate}
\end{enumerate}

The formal construction of the expansion mapping follows. The inductive
hypothesis in this construction of
$\ex: \plays(\game)\to\plays(\game(p, \lambda))$ is the following:
for all finite plays $\pi \in \mathsf{Plays}_\mathit{fin}$, using $(\ell,\nu)$ to denote $\last(\pi)$ and $(\ell',q,\bar{\nu})$ to denote
$\mathsf{last}(\mathsf{Ex}(\pi))$, then
$\ell' = \ell$, $q \neq \mathsf{bad}$ and $\bar{\nu}_{|C} = \nu$.
We proceed by induction on the number of moves along a play.

The base case consists of plays of $\game$ with no moves, i.e.,
plays in which there is a single state.
For any play $(\ell, \nu)$, we set
$\mathsf{Ex}((\ell, \nu))$ to be the play $(\ell, p(\ell), \bar{\nu})$ of
$\game(p, \lambda)$ consisting of a single state, where
$\bar{\nu}_{|C} = \nu$ and $\bar{\nu}(z)=0$. The inductive hypothesis is
verified: the states $(\ell, \nu)$ and $(\ell, p(\ell), \bar{\nu})$ share the
same location of $\automaton$, their clock valuations agree over $C$ and
$(\ell, p(\ell))$ is not a bad location.

Next we assume that expansions are defined for all plays with $n$ moves.
Fix $\pi = s_0(m_0^{(1)}, m_0^{(2)})\ldots s_{n+1}$
a play with $n+1$ moves and assume the expansion of its prefix
$\pi_{|n} = s_0(m_0^{(1)}, m_0^{(2)})\ldots s_{n}$ has already been constructed.
Write $s = (\ell, \nu) = \last(\pi_{|n})$, $s' = (\ell', \nu') = \last(\pi)$,
$\bar{s} = (\ell, q, \bar{\nu}) = \last(\ex(\pi_{|n}))$ and for $i\in\{1, 2\}$,
$m_n^{(i)} = (d_n^{(i)}, a_n^{(i)})$. We assume w.l.o.g.~that
$s\xrightarrow{m_n^{(1)}}s'$ holds: the induction step can be done similarly
by exchanging the roles of the players if $s\xrightarrow{m_n^{(1)}}s'$ does
not hold. This assumption implies $\player_1$ is faster or as fast
as $\player_2$. If $\player_2$ were strictly faster, then
$s\xrightarrow{m_n^{(2)}}s'$ would hold, in turn implying
that $\nu'(\gamma) = \nu(\gamma) + d_n^{(2)}$ ($\gamma$ cannot be reset).
However, since $s\xrightarrow{m_n^{(1)}}s'$ holds, it follows that
$\nu'(\gamma) = \nu(\gamma) + d_n^{(1)}$, contradicting the assumption
that $\player_2$ is faster. In other words, we must have
$d_n^{(1)}\leq d_n^{(2)}$.
We separate the construction in multiple cases.
\begin{enumerate}
\item \label{expansion:q:even}
  If $q$ is even, the moves $m_n^{(1)}$ and $m_n^{(2)}$ are enabled in
  $\bar{s}$ by construction. Indeed, $\nu$ and $\bar{\nu}$ agree over $C$, and
  we have $I(\ell) = I'((\ell, q))$ and for any
  outgoing edge $(\ell, g, a, D, \tilde{\ell})$ of $\ell$ in $\automaton$
  there is an edge
  $((\ell, q), g, a, D\cup\{z\}, (\tilde{\ell}, p(\tilde{\ell})))$
  in $\automaton(p, \lambda)$ with the same guard and which resets $z$.
  We distinguish cases following whether a delay is taken
  or not.

  If $m^{(1)}_n$ is a delay move, we set $\mathsf{Ex}(\pi) =
  \mathsf{Ex}(\pi_{|n})(m_n^{(1)}, m_n^{(2)}) (\ell, q, \bar{\nu} + d)$. This
  play is well-defined: $q$ is even, thus $\ell$ and $(\ell, q)$ share the
  same invariant and support the same delay moves.
  The inductive hypothesis is verified: $\ell' = \ell$ because
  a transition labeled by a delay move was taken, $q\neq \bad$ and
  $\nu' = \nu + d = \bar{\nu}_{|C} + d$.
  
  Otherwise, $m^{(1)}_n$ is not a delay move and is associated with an edge of
  the TA. We set
  $\mathsf{Ex}(\pi) =
  \mathsf{Ex}(\pi_{|n})(m_n^{(1)}, m_n^{(2)}) (\ell', p(\ell'), \bar{\nu}')$
  with $\bar{\nu}'_{|C} = \nu'$ and $\bar{\nu}'(z) = 0$.
  This is a well-defined play, owing to the edges recalled above.
  It is not difficult to verify that the inductive hypothesis is satisfied.

  Note that $z$ is reset in the second case. It may be the
  case that $\player_2$ is not responsible for the last
  transition in the expansion despite being responsible for the last
  transition of $\pi$.
  In other words, it is possible that
  both $s\xrightarrow{m_n^{(1)}}s'$ and $s\xrightarrow{m_n^{(2)}}s'$
  hold, but that $\bar{s}\xrightarrow{m _n^{(2)}}\last(\ex(\pi))$ does not
  hold. This occurs whenever the moves of both players share the same target
  state in the base TG
  but one player uses a delay move and the other a move with an action. We
  choose to have $\player_1$'s move be responsible for the transition in the
  expansion in this case.
  This choice is for technical reasons related to blamelessness.
\item \label{expansion:q:odd}
  If $q$ is odd, one or both of the moves $m^{(1)}_n$ or $m^{(2)}_n$ may not
  be enabled
  in $\bar{s}$ due to the different invariant.
  Recall that $I'((\ell, q))=(I(\ell)\land z\leq \lambda)$ and for any
  outgoing edge $(\ell, g, a, D, \tilde{\ell})$ of $\ell$ in $\automaton$
  there is an edge $((\ell, q), (g\land z<\lambda), a, D, (\tilde{\ell},
  \min\{q, p(\tilde{\ell})\}))$  in $\automaton(p, \lambda)$.
  If $a\in\Sigma$, a move $(t, a)$ is disabled in state $\bar{s}$ if
  $\bar{\nu}(z) + t \geq \lambda$. We reason as follows, depending on whether
  a delay of $d$ allows clock $z$ to reach $\lambda$ from the state $\bar{s}$.
  \begin{enumerate}
  \item \label{expansion:q:odd:sub1}
    Assume $\bar{\nu}(z) + d < \lambda$. Then $m_n^{(1)}$ is
    enabled in $\bar{s}$ (recall we assume $\player_1$ is responsible for the
    last transition of $\pi$). 
    To ensure the $\player_2$-selected move in the expansion is
    enabled in $\bar{s}$, we alter $m_n^{(2)}$ it if its delay is too large:
    let $\tilde{m}_n^{(2)} = m_n^{(2)}$ if $\bar{\nu}(z) + d_n^{(2)} < \lambda$
    and $\tilde{m}_n^{(2)}= (\lambda - \bar{\nu}(z), \beta_2)$ otherwise.

    If $m_n^{(1)}$ is a delay move, define
    $\ex(\pi) =
    \ex(\pi_{|n})(m_n^{(1)}, \tilde{m}_n^{(2)})(\ell, q, \bar{\nu} + d)$.
    This is a well-defined play: $\bar{\nu} + d \models I'((\ell, q))$ holds
    because $I'((\ell, q))=(I(\ell)\land z\leq\lambda)$, $\nu$ and
    $\bar{\nu}$ coincide on $C$ and the move $m_n^{(1)}$ is available in $s$, 
    and $\bar{\nu}(z) + d < \lambda$.
    Otherwise, if $m_n^{(1)}$ is not a delay move, define
    $\ex(\pi) = \ex(\pi_{|n})(m_n^{(1)}, \tilde{m}_n^{(2)})
    (\ell', \min\{q, p(\ell')\}, \bar{\nu}')$ where $\bar{\nu}'_{|C} = \nu'$
    and $\bar{\nu}'(z) = \bar{\nu}(z) + d$. By definition of the edges
    recalled above, and because $\nu$ and $\bar{\nu}$ coincide on $C$ and the
    move $m_n^{(1)}$ is available in $s$, we conclude that $\mathsf{Ex}(\pi)$
    is a well-defined play.
    In either case, the inductive hypothesis is satisfied.
    
  \item Otherwise, assume $\bar{\nu}(z) + d \geq \lambda$. In this case,
    a bad location appears along $\ex(\pi)$.
    Denote by $t=\lambda-\bar{\nu}(z)$ the time left before the
    current window becomes a bad window. For any non-negative real $r$, we
    write $\bar{\nu}_r= \mathsf{reset}_{\{z\}}(\bar{\nu}+r)$ for the
    clock valuation obtained by shifting $\bar{\nu}$ by $r$ time units
    and then resetting $z$, $b^{(i)}_r$ for the move
    $(r, \beta_i)$ and $m_n^{(i)}-r$ for  $(d_n^{(i)}-r, a_n^{(i)})$, i.e., the
    move $m_n^{(i)}$ with a delay shortened by $r$ time units.

    Recall, when $(\ell, \bad)$ is left, by definition of edges
    of $\automaton(p, \lambda)$, the location
    $(\ell, p(\ell))$ is entered. Depending on the parity of $p(\ell)$, the
    invariant of $(\ell, p(\ell))$ is different.
    Thus, there are two cases to consider: $p(\ell)$ is even and
    $p(\ell)$ is odd.

    \begin{enumerate}
    \item \label{expansion:q:odd:subsub1} If $p(\ell)$ is even, we set
      $\mathsf{Ex}(\pi)$ to be the play
    \begin{equation}\label{eqn:expansion:even}
      \mathsf{Ex}(\pi_{|n})
      (b_t^{(1)}, b_t^{(2)} )
      (\ell, \bad, \bar{\nu}_t)
      (b_0^{(1)}, b_0^{(2)} )
      (\ell, p(\ell), \bar{\nu}_t)
      (m_n^{(1)}-t, m_n^{(2)}-t)
      (\ell', p(\ell'), \bar{\nu}'),
    \end{equation}
    where $\bar{\nu}'_{|C}=\nu'$ and $\bar{\nu}'(z)=0$ if $a_n^{(1)}\in\Sigma_1$
    and $\bar{\nu}'(z)=d-t$ otherwise ($z$ is reset if $\player_1$'s action is
    not a delay). We obtain this expansion by first using actions $\beta_i$
    with the delay $t$ to enter a bad location, then the actions
    $\beta_i$ immediately again to exit the bad location and finally use
    the original moves of the players, but with an offset of $t$ times
    units ($t$ time units passed before entering the bad location).
    This expansion is a play of
    $\game(p, \lambda)$: the moves $b_t^{(1)}$ and $b_t^{(2)}$ are enabled in
    $(\ell, q, \bar{\nu})$ as
    $\bar{\nu}+t\models I'((\ell, q))$ (because $\nu + d\models I(\ell)$ as
    $m_n^{(1)}$ is enabled in $(\ell, \nu)$, $\bar{\nu}(z) + t = \lambda$ and
    $I'((\ell, q))=(I(\ell)\land z\leq\lambda)$), and lead to the
    state $(\ell, \bad, \bar{\nu}_t)$ (recall edges entering bad locations
    reset $z$).
    The moves $b_0^{(1)}$ and $b_0^{(2)}$ are enabled in
    $(\ell, \bad, \bar{\nu}_t)$ due to the edges
    $((\ell, \bad), \true, \beta_i, \varnothing, (\ell, p(\ell)))$ for
    $i\in\{1, 2\}$ of
    $\automaton(p, \lambda)$. One can argue the moves $m_n^{(1)}-t$ and
    $m_n^{(2)}-t$ are enabled in $(\ell, p(\ell), \bar{\nu}_t)$ using
    the same arguments as in case \ref{expansion:q:even}.
    The inductive hypothesis is preserved in this
    case.

  \item Whenever $p(\ell)$ is odd, the invariant of $(\ell, p(\ell))$
    implies $z\leq \lambda$.
    Let $\mu$ denote the integral part of
    $\frac{\bar{\nu}(z) + d}{\lambda}$; $\mu$ represents
    the number of bad windows we detect with our construction during delay
    $d$. In a nutshell,
    we divide the single step in the original TG into
    $2\mu+1$ steps in the expansion: we enter and exit the bad location
    $\mu$ times and finally play the original move in the end with the
    remaining delay. It may be necessary to modify $\player_2$'s move due to
    the invariant of the odd location $(\ell, p(\ell))$ implying
    $z \leq \lambda$ and the guard of its outgoing edges labeled by actions in
    $\Sigma$ implying $z<\lambda$.
    To this end, let $\tilde{m}_n^{(2)} = m_n^{(2)}$ if
    $d_n^{(2)} - t < \mu\lambda$ and
    $\tilde{m}_n^{(2)} = (\mu\lambda + t, \beta_2)$ otherwise.
    If $\mu=1$, we get an expansion similar to the previous case. We set
     $\mathsf{Ex}(\pi)$ to be 
     \begin{equation*}
       \mathsf{Ex}(\pi_{|n})
      (b_t^{(1)}, b_t^{(2)} )
      (\ell, \bad, \bar{\nu}_t)
      (b_0^{(1)}, b_0^{(2)} )
      (\ell, p(\ell), \bar{\nu}_t)
      (m_n^{(1)}-t, \tilde{m}_n^{(2)}-t)
      (\ell', q', \bar{\nu}'),
    \end{equation*}
    where $q'= \min\{p(\ell), p(\ell')\}$,
    $\bar{\nu}'_{|C}=\nu'$ and $\bar{\nu}'(z) = d-t$. That is
    indeed a play of $\game(p, \lambda)$: moves $b^{(i)}_r$ ($i\in\{1, 2\}$
    and $r\in\{0, t\}$) are enabled for the same reasons as the previous case
    and the other moves are valid for the same reasons as in
    case 2.a: $d-t = d - \lambda +\bar{\nu}(z) < \lambda$ because $\mu =1$.
    For $\mu\geq 2$, we define $\mathsf{Ex}(\pi)$ as
    \begin{equation*}
    \begin{array}{l}
      \mathsf{Ex}(\pi_{|n})
      (b_t^{(1)}, b_t^{(2)} )
      (\ell, \bad, \bar{\nu}_t) 
      (b_0^{(1)}, b_0^{(2)} )  
      \\
      \left.
        \begin{array}{l}
          (\ell, p(\ell), \bar{\nu}_t)
          (b_\lambda^{(1)}, b_\lambda^{(2)} )
          (\ell, \bad, \bar{\nu}_{t+\lambda})
          (b_0^{(1)}, b_0^{(2)} )
          \\\phantom{aaaaaaaaaaaaaaaaaaaaa}\vdots
          \\
          (\ell, p(\ell), \bar{\nu}_{t+(\mu-2)\lambda})
          (b_\lambda^{(1)}, b_\lambda^{(2)} )
          (\ell, \bad, \bar{\nu}_{t+(\mu-1)\lambda})
          (b_0^{(1)}, b_0^{(2)} )
        \end{array}
        \right\} \mu-1 \text{ steps}
      \\
      (\ell, p(\ell), \bar{\nu}_{t'})(m_n^{(1)}-t', \tilde{m}_n^{(2)}-t')
      (\ell', q', \bar{\nu}'),
    \end{array}
  \end{equation*}
    where $q'=\min\{p(\ell), p(\ell')\}$, $t'= t + (\mu-1)\lambda$ represents
    the time spent repeatedly  entering and exiting the bad location,
    $\bar{\nu}'_{|C}=\nu'$ and $\bar{\nu}'(z) = d - t'$. This is a well-defined
    play of $\game(p, \lambda)$ for the reasons argued above. Extending the
    expansion this way preserves the inductive hypothesis.
    \end{enumerate}
  \end{enumerate}
\end{enumerate}

This construction generalizes to infinite plays.
In some cases, an expansion of a finite play  may contain more steps. However,
this is only the case when a bad location appears in an expansion. If a finite play
and its expansion share the same number of steps, we say that they are
\textit{coherent}.
The expansion mapping preserves time-convergence and divergence for infinite
paths: the sum of delays are identical in a play and its expansion.

The behavior of the expansion mapping
w.r.t.~suffixes and well-initialized plays is of interest for studying
the connection between the non-direct
timed window  objective on the base TG and the co-Büchi objective on
the expanded TG. Given a (finite or infinite) play
$\pi= (\ell_0, \nu_0)(m_0^{(1)}, m_0^{(2)})\ldots$,
a suffix of $\mathsf{Ex}(\pi)$ is not necessarily the
expansion of a suffix of $\pi$. However, any
well-initialized suffix of $\mathsf{Ex}(\pi)$ can be shown to be the expansion
of a delayed suffix $\pi_{n\to}^{+d}$ for some $n\in\IN$ and
$d\in[0, \delay(m_n^{(1)}, m_n^{(2)})]$.
A well-initialized suffix of an expansion starts whenever an even location
is left through an edge (case 1) or a bad location is left through an edge
(case 2.b). In the former case, the suffix of the expansion is the expansion
of a suffix $\pi_{n\to}$ by construction (the last state of the expansion,
viewed as a play, is well-initialized). In the latter case, a delayed suffix
may be required. This is observable with equation \eqref{eqn:expansion:even}
(and is similar in other cases involving bad locations):
the suffix
$(\ell, p(\ell), \bar{\nu}_t)(m_n^{(1)}-t, m_n^{(2)}-t)(\ell', p(\ell'), \bar{\nu}')$
of the expansion under construction
is the expansion of the shifted suffix $\pi_{n\to}^{+t}$ of the finite play
under consideration and is well-initialized.

\subsubsection*{Projection mapping} The counterpart to the expansion mapping
is the projection mapping
$\mathsf{Pr}\colon \plays(\game(p, \lambda)) \to \plays(\game)$.
The projection mapping
removes window information in any play in $\game(p, \lambda)$ to obtain
a play in $\game$. Any action $\beta_1$ or $\beta_2$ is replaced by
the action $\bot$.
Formally, we define the projection mapping over finite and infinite plays as
follows.

For any (finite or infinite) play
$\bar{\pi} = (\ell_0, q_0, \bar{\nu}_0)
((d_0^{(1)}, a_0^{(1)}), (d_0^{(2)}, a_0^{(2)}))\ldots \in
\plays(\game(p, \lambda))$, we set $\pr(\bar{\pi})$ to be the sequence
$\pi = (\ell_0, (\bar{\nu}_0)_{|C}) (\tilde{m}_0^{(1)}, \tilde{m}_0^{(2)})\ldots$ where for all $i\in\{1, 2\}$ and all $j$,
$\tilde{m}_j^{(i)}=(d_j^{(i)}, a_j^{(i)})$ if $a_j^{(i)}\notin\{\beta_1, \beta_2\}$
and $\tilde{m}_j^{(i)}=(d_j^{(i)}, \bot)$ otherwise. This sequence is indeed
a well-defined play: any move $(d, a)$ enabled in a state
$\bar{s}_j = (\ell_j, q_j, \bar{\nu}_j)$ such that
$a\notin\{\beta_1, \beta_2\}$ is enabled in
$s_j = (\ell_j, (\bar{\nu}_j)_{|C})$. If the expanded location of
$\bar{s}_j$ is bad, the only such move is $(0, \bot)$. If it is not bad,
guards of outgoing edges and invariants in $\automaton(p, \lambda)$ are
either the same or strengthened from their counterpart
in $\automaton$, i.e., if the constraints are
verified in the expanded TA, they must be in the original one.
Furthermore, as edges unrelated to bad states are derived from the edges of
the original TA, this ensures that
$s_j\in\delta(s_{j-1}, \tilde{m}_{j-1}^{(1)}, \tilde{m}_{j-1}^{(2)})$
for all $j>0$ (where $\delta$ is the joint destination function).

The projection mapping preserves time-divergence. Unlike the expansion mapping,
projecting a finite play does not alter the amount of moves.
This mapping respects suffixes: for all finite plays $\bar{\pi}$
and (finite or infinite) plays $\bar{\pi}'$ of $\game(p, \lambda)$ such
that $\bar{\pi}((d^{(1)}, a^{(1)}), (d^{(2)}, a^{(2)}))\bar{\pi}'$
is a well-defined play, we have
$\mathsf{Pr}(\bar{\pi}(d^{(1)}, a^{(1)}), (d^{(2)}, a^{(2)})\bar{\pi}')=
\mathsf{Pr}(\bar{\pi})(\tilde{m}^{(1)}, \tilde{m}^{(2)})\mathsf{Pr}(\bar{\pi}')$,
where the move $\tilde{m}^{(i)}$ is $(d^{(i)}, a^{(i)})$ if
$a^{(i)}\notin\{\beta_1, \beta_2\}$
and $\tilde{m}^{(i)}=(d^{(i)}, \bot)$ otherwise.
We refer to this property as \textit{suffix compatibility}.

\subsubsection*{Objective preservation}
We now establish the main theorem of this section: a play of $\game$ satisfies
the (resp.~direct) timed window objective if and only if its expansion
satisfies the co-Büchi (resp.~safety) objective over bad locations; and
a play of $\game(p, \lambda)$ satisfies the co-Büchi (resp.~safety) objective
over bad locations if and only if its projection satisfies the (resp.~direct)
timed window objective.

\begin{lemma}\label{theorem:window_mapping}
  The following assertions hold.
  For all time-divergent plays $\pi\in\plays_\infty(\game)$:
  \begin{enumerate}[{A}.1.]
  \item $\pi\in\mathsf{DTW}(p, \lambda)$
    if and only if $\mathsf{Ex}(\pi)\in\mathsf{Safe}(\badset)$;
    \label{item:expansion:1}
  \item $\pi\in\mathsf{TW}(p, \lambda)$
    if and only if $\mathsf{Ex}(\pi)\in\mathsf{coB\ddot{u}chi}(\badset)$.
    \label{item:expansion:2}
  \end{enumerate}

  For all well-initialized time-divergent plays
  $\bar{\pi}\in\plays_\infty(\game(p, \lambda))$:
  \begin{enumerate}[{B}.1.]
  \item 
    $\bar{\pi}\in\mathsf{Safe}(\badset)$
    if and only if $\mathsf{Pr}(\bar{\pi})\in\mathsf{DTW}(p, \lambda)$;
    \label{item:projection:1}
  \item
    $\bar{\pi}\in\mathsf{coB\ddot{u}chi}(\badset)$
    if and only if $\mathsf{Pr}(\bar{\pi})\in\mathsf{TW}(p, \lambda)$.
    \label{item:projection:2}    
  \end{enumerate}
\end{lemma}

The form of this result is due to the lack of a
bijection between the sets of plays of a TG and its expansion
(Remark~\ref{remark:no_bijection}). 
This lemma essentially follows from the construction of $\automaton(p, \lambda)$
and the definitions of the expansion and projection mappings.

\begin{proof}
  We start with \textbf{Item A.\ref{item:expansion:1}}. Fix a time-divergent play
  $\pi=(\ell_0, \nu_0)(m_0^{(1)}, m_0^{(2)})\ldots\in\plays_\infty(\game)$
  and write
  $\bar{\pi}=(\ell'_0, q_0, \bar{\nu}_0)(\tilde{m}_0^{(1)}, \tilde{m}_0^{(2)})\ldots$ for its expansion. Assume $\pi$ satisfies
  the direct timed window  objective. We establish that $\bar{\pi}$ is
  safe and all of its prefixes are coherent with their expansion
  using an inductive argument. Intuitively,
  for the first step, we show that if the window opened at step $0$ closes at
  step $n$, then the $n$th location of $\bar{\pi}$ is even and is reached in
  less than $\lambda$ time units.
  
  Since $\pi\in\mathsf{TGW}(p, \lambda)$, there is
  some $n_0$ such that $\nu_{n_0}(\gamma)-\nu_0(\gamma)<\lambda$,
  $\min_{0\leq k < j}p(\ell_k)$ is odd for all $0\leq j <n_0$ and
  $\min_{0\leq k\leq n_0}p(\ell_k)$ is even.
  Since $\bar{\nu}_0(z)=0$, $\bar{\pi}$ and $\pi$ are coherent
  up to step $n_0$ as $z$ does not reach $\lambda$ (i.e., no bad locations can
  occur up to step $n_0$).
  By construction of the expansion mapping (see
  case \ref{expansion:q:odd} of the construction),
  $q_{n_0} =\min_{0\leq k\leq n_0}p(\ell_k)= p(\ell_{n_0})$ and $p(\ell_{n_0})$
  is even.
  From there, there are two possibilities: either location
  $\ell_{n_0}$ is never left (only delays are taken from there) or
  after some delay moves, location $\ell_{n_0}$ is exited. In the first case,
  the safety objective is trivially satisfied as the location
  $(\ell_{n_0}, p(\ell_{n_0}))$ is never left in the expansion and
  no bad locations were visited beforehand.
  Otherwise, some edge is traversed (for the first time since step $n_0$)
  at some step $j_0$.
  Then, $\bar{\nu}_{j_0+1}(z) = 0$, $q_{j_0+1} = p(\ell_{j_0+1})$ and no bad
  locations appear in the first $j_0$ steps of $\bar{\pi}$.

  We can repeat the first argument in position $j_0+1$ because
  $\pi\in\mathsf{DTW}(p, \lambda)$ and thus
  $\pi_{j_0+1\to}\in\mathsf{TGW}(p, \lambda)$.
  We conclude there is some
  $n_1\geq j_0+1$ such that $q_n\neq \bad$ for $j_0+1\leq n\leq n_1$ and
  $q_{n_1}$ is even. Once more, we separate cases following if $\ell_{n_1}$
  is left through an edge.
  Iterating this argument shows that no bad locations appear along $\bar{\pi}$.

  Conversely, assume $\pi$ does not satisfy the direct window  objective.
  Let $j_0$ be the smallest index $j$ such that
  $\pi_{j\to}\notin\mathsf{TGW}(p, \lambda)$ (Lemma~\ref{lemma:discrete_twpo}).
  We can argue, using a similar
  inductive argument as above, that $\pi$ and $\bar{\pi}$ are
  coherent up to step $j_0$, no bad locations occur up to
  step $j_0$, and that
  $q_{j_0} = p(\ell_{j_0})$ and $\bar{\nu}_{j_0}(z) = 0$. In the sequel, we argue
  that a bad location is entered using the fact that there is no good window
  at step $j_0$.
  
  The negation of $\mathsf{TGW}(p, \lambda)$ yields that for all
  $j\geq j_0$, if $\nu_{j}(\gamma)-\nu_{j_0}(\gamma)<\lambda$, then
  $\min_{j_0\leq k\leq j}p(\ell_k)$ is odd.
  There is some $j$ such that $\nu_{j}(\gamma)-\nu_{j_0}(\gamma)\geq\lambda$
  as $\pi$ is assumed to be time-divergent. Write $j_1$ for the
  smallest such $j$. As $j_1$ is minimal, 
  $\nu_{j_1-1}(\gamma)-\nu_{j_0}(\gamma) <\lambda$ holds.
  It follows from the above that
  $q_{j_1-1}$ is odd as $q_{j_1-1}=\min_{j_0\leq k\leq j_1}p(\ell_k)$ and that there are no resets of $z$ between steps
  $j_0$ and $j_1-1$ (resets of $z$ require an even location or a bad location),
  hence $\bar{\nu}_{j_1-1}(z) = \nu_{j_1-1}(\gamma)-\nu_{j_0}(\gamma)$.
  The delay $d_{j_1-1}= \delay(m_{j_1-1}^{(1)}, m_{j_1-1}^{(2)})$ is such that
  $\bar{\nu}_{j_1-1}(z) + d_{j_1-1} =\nu_{j_1}(\gamma)-\nu_{j_0}(\gamma)\geq\lambda$;
  the definition of the expansion function redirects $\bar{\pi}$ to a
  bad location (case~\ref{expansion:q:odd}.b of the expansion definition).
  This shows that $\bar{\pi}$ does not satisfy $\safe(\badset)$.

  Let us move on to \textbf{Item~A.\ref{item:expansion:2}}. Fix a time-divergent play
  $\pi=(\ell_0, \nu_0)(m_0^{(1)}, m_0^{(2)})\ldots\in\plays_\infty(\game)$
  and write
  $\bar{\pi} = (\ell_0', q_0, \bar{\nu}_0)(\tilde{m}_0^{(1)}, \tilde{m}_0^{(2)})\ldots$ for
  its expansion.
  Assume $\bar{\pi}$ does not
  satisfy the objective $\mathsf{coB\ddot{u}chi}(\badset)$, i.e.,
  there are infinitely many occurrences
  of bad locations along $\bar{\pi}$. We establish that
  $\pi\notin\mathsf{TW}(p, \lambda)$, using Item~A.\ref{item:expansion:1}.

  It suffices to show that for infinitely
  many $n$, there is some $d$ such that
  $\pi_{n\to}^{+d}\notin\mathsf{DTW}(p, \lambda)$. Indeed,
  Lemma~\ref{lemma:discrete_twpo} and this last assertion imply that
  for infinitely many $n$, $\pi_{n\to}\notin\mathsf{DTW}(p, \lambda)$, hence,
  there are infinitely many bad windows along $\pi$.
  In turn, this property ensures no suffix of $\pi$ satisfies the direct timed
  window  objective, i.e., $\pi\notin\mathsf{TW}(p, \lambda)$.
  
  Recall a well-initialized suffix of $\bar{\pi}$ is the expansion of some
  $\pi_{n\to}^{+d}$ and a well-initialized suffix always follows after a bad
  location. Bad locations are assumed to occur infinitely often
  along $\bar{\pi}$, therefore $\bar{\pi}$ has infinitely many
  well-initialized suffixes. It follows from $\bar{\pi}\in\cobuchi(\badset)$
  that there are infinitely many
  well-initialized suffixes of $\bar{\pi}$ that violate $\safe(\badset)$.
  Therefore, there are infinitely
  many $n$ such that $\ex(\pi_{n\to}^{+d})$ does not satisfy $\safe(\badset)$
  for some $d$, because
  each step of $\pi$ induces finitely many visits to a bad
  location in the expansion by construction (in other words, there are
  finitely many well-initialized suffixes of $\bar{\pi}$ per step in $\pi$).
  From Item~A.\ref{item:expansion:1}, there are infinitely many $n$
  such that $\pi_{n\to}^{+d}$ does not satisfy the direct window  objective
  for some $d$. This ends this direction of the proof.

  Conversely, assume $\bar{\pi}$ satisfies the objective
  $\mathsf{coB\ddot{u}chi}(\badset)$. If $\bar{\pi}$ is safe, then
  it follows from Item~A.\ref{item:expansion:1} that
  $\pi\in\mathsf{DTW}(p, \lambda)\subseteq \mathsf{TW}(p, \lambda)$.
  If $\bar{\pi}$ is unsafe, there is a well-initialized suffix of
  $\bar{\pi}$ that is safe, i.e., the suffix of $\bar{\pi}$  starts after
  the last occurrence of a bad location. This suffix is of the form
  $\ex(\pi_{n\to}^{+d})$ for some $n$ and $d$. By
  Item~A.\ref{item:expansion:1},
  $\pi_{n\to}^{+d}\in\mathsf{DTW}(p, \lambda)$, which implies
  that $\pi_{(n+1)\to}\in\mathsf{DTW}(p, \lambda)$. Thus, we have
  $\pi\in\mathsf{TW}(p, \lambda)$, which ends the proof of this item.
  
  Let us proceed to \textbf{Item B.\ref{item:projection:1}}.
  Let $\bar{\pi} = (\ell_0, q_0, \bar{\nu}_0)(m_0^{(1)}, m_0^{(2)})\cdots
  \in\plays_\infty(\game(p, \lambda))$ be a time-divergent
  well-initialized play and write $\pi$ for its projection.
  Assume $\bar{\pi}$ satisfies the objective $\mathsf{Safe}(\badset)$ and
  let us prove that $\pi\in\mathsf{DTW}(p, \lambda)$.
  It suffices to establish that for all $n\in\IN$,
  $\pi_{n\to}\in \mathsf{TGW}(p, \lambda)$ (Lemma~\ref{lemma:discrete_twpo}).
  First, we argue there is an even location along $\bar{\pi}$. Then
  we establish that when the first even location is entered in $\bar{\pi}$,
  the window opened at step 0 in $\pi$ closes.
  We conclude, using the inductive property of windows and an inductive
  argument, that $\pi$ satisfies the direct window objective.
  
  Safety and divergence of $\bar{\pi}$ ensure that an even location appears
  along $\bar{\pi}$.
  Assume that there are no even locations along $\bar{\pi}$. Then
  every location appearing in $\bar{\pi}$ must be odd by safety of $\bar{\pi}$.
  Thus $z$ cannot be reset as it requires
  exiting an even location or entering a bad location.
  The invariant $z\leq \lambda$ of odd locations would prevent
  time-divergence of $\bar{\pi}$, which would be contradictory.
  Thus, there must be some even location along $\bar{\pi}$.

  Let $n_0$ denote the smallest index $n$ such that $q_n$ is even. We establish
  that the window opened at step $0$ in $\pi$ closes at step $n_0$.
  It must hold that
  $\bar{\nu}_{n_0}(z) = \bar{\nu}_{n_0}(\gamma)-\bar{\nu}_{0}(\gamma) < \lambda$,
  as the outgoing edges of odd locations that do not target bad locations
  have guards implying $z<\lambda$ and do not reset $z$.
  Furthermore, $q_n$ is odd for all $n<n_0$
  and $q_{n_0} = \min_{0\leq k\leq n_0} p(\ell_k)$, by
  definition of the edge relation of $\automaton(p, \lambda)$
  and since $q_0=p(\ell_0)$. This proves that the window
  opened at step 0 closes at step $n_0$ in less than $\lambda$ time units.
  It follows from the inductive
  property of windows (Lemma~\ref{lemma:inductive_prop}) that for all
  $0\leq n\leq n_0$, $\pi_{n\to}\in\mathsf{TGW}(\lambda, p)$.
  
  As $q_n$ is odd for $n< n_0$, we have $p(\ell_{n_0}) = q_{n_0}$.
  There are two possibilities for $\bar{\pi}$: either the even location
  $(\ell_{n_0}, p(\ell_{n_0}))$ is never left or
  there is some $j_0$ such that at step $j_0$, the expanded location is
  exited through an edge (via the pair of moves $(m_{j_0}^{(1)}, m_{j_0}^{(2)})$).
  In the first case, only delays are taken in the location
  $\ell_{n_0}$ in $\pi$, the priority of which is even, yielding
  $\pi_{n\to}\in\mathsf{TGW}(p, \lambda)$ for all $n\geq n_0$ (all windows
  after step $n_0$ close immediately), and combining
  this with the previous paragraph implies $\pi\in\mathsf{DTW}(p, \lambda)$.
  In the latter case, we have $\pi_{n\to}\in\mathsf{TGW}(p, \lambda)$ for
  $n_0\leq n\leq j_0$ (similarly to the former case, only delays are taken
  in $\ell_{n_0}$ in $\pi$ up to step $j_0$), and $\bar{\nu}_{j_0+1}(z)=0$ and
  $q_{j_0+1}=p(\ell_{j_0+1})$ as edges leaving even locations reset $z$ and
  lead to locations of the form $(\ell', p(\ell'))$.
  Repeating the previous arguments from position
  $j_0+1$ ($\bar{\pi}_{j_0+1\to}$ is a time-divergent well-initialized suffix of
  $\bar{\pi}$),
  one can find some $n_1$ such that $q_{n_1}=p(\ell_{n_1})$ is even
  and for all $j_0+1\leq n\leq n_1$, $\pi_{n\to}\in\mathsf{TGW}(p, \lambda)$.
  Once more,
  we split in cases following whether any edge is traversed in location
  $(\ell_{n_1}, p(\ell_{n_1}))$. It follows from an induction that
  $\pi\in\mathsf{DTW}(p, \lambda)$.

  Assume now that $\bar{\pi}$ does not satisfy the safety objective. There is
  some smallest $n$ such that $q_n=\bad$. Let $j_0$ be $0$ if $z$ was never
  reset before position $n$ or the greatest $j < n$ such that the
  pair of moves $(m_{j-1}^{(1)}, m_{j-1}^{(2)})$ induces a reset of $z$ otherwise.
  We have $q_{j_0}=p(\ell_{j_0})$ and $\bar{\nu}_{j_0}(z)=0$: if $j_0=0$, this
  is due to $\bar{\pi}$ being well-initialized and if $j_0>0$, this follows
  from $n$ being the smallest index of a bad location and from the
  definition of edges in $\automaton(p, \lambda)$;
  edges that reset $z$ have as their target expanded locations of the
  form $(\ell,p(\ell))$ or $(\ell,\bad)$.
  We argue that $\pi_{j_0\to}$ does not satisfy the timed good window objective.
  There cannot be any even locations between positions $j_0$ and $n$; there are
  no edges to bad locations in even locations and edges leaving even locations
  reset $z$. Thus, for all $j_0\leq j<n$, $q_j=\min_{j_0\leq k\leq j}p(\ell_k)$
  and $q_j$ is odd. Bad locations can only be entered when $z$ reaches $\lambda$
  in an odd location: the window opened at step $j_0$ does not close in time.
  Therefore, $\pi\notin\mathsf{DTW}(p, \lambda)$.

  For \textbf{Item B.\ref{item:projection:2}}, we use compatibility of the projection
  mapping with suffixes. Fix a divergent well-initialized play
  $\bar{\pi} \in\plays_\infty(\game(p, \lambda))$ and let $\pi$ be its
  projection.
  Assume $\bar{\pi}$ satisfies the objective
  $\mathsf{coB\ddot{u}chi}(\badset)$. If $\bar{\pi}$ is safe, we have the property
  by Item B.\ref{item:projection:1}. Assume some bad location appears along
  $\bar{\pi}$. As the co-Büchi objective is satisfied, this location is left
  and the suffix following this exit is well-initialized.
  As there are finitely many bad locations along $\bar{\pi}$, it follows
  $\bar{\pi}$ has a well-initialized suffix satisfying $\safe(\badset)$.
  From the compatibility of projections with suffixes and
  Item B.\ref{item:projection:1},
  $\pi$ has a suffix satisfying the direct timed window  objective,
  hence $\pi$ satisfies the timed window  objective.

  Conversely, assume $\bar{\pi}$ does not satisfy
  $\mathsf{coB\ddot{u}chi}(\badset)$, i.e., there are infinitely many occurrences
  of bad locations along $\bar{\pi}$. Divergence ensures any bad location is
  eventually left through an edge, yielding a well-initialized suffix.
  From Item
  B.\ref{item:projection:1} and suffix compatibility of the projection mapping,
  it follows that $\pi$ has infinitely many suffixes that do not satisfy the
  direct timed window  objective. This implies that $\pi$ does not satisfy
  the timed window  objective, ending the proof.
\end{proof}

This lemma completely disregards plays that are not well-initialized.
This is not an issue, as any play starting in the initial state of the TG
$\game(p, \lambda)$ is well-initialized. Indeed, if $\ell_0$ is the initial
location of $\automaton$, then $(\ell_0, p(\ell_0))$ is the initial location of
$\automaton(p, \lambda)$, and thus $(\ell_0, p(\ell_0), \mathbf{0}^{C\cup\{z\}})$
is the initial state of $\mathcal{T}(\automaton(p, \lambda))$.
Any initial play of $\game$
expands to an initial play of $\game(p, \lambda)$ and any initial play of
$\game(p, \lambda)$ projects to an initial play of $\game$.

The previous result can be leveraged to prove that the verification problem for
the (resp.~direct) timed window  objective on $\automaton$ can
be reduced to the verification problem for the co-Büchi (resp.~safety) objective
on $\automaton(p, \lambda)$.

\begin{theorem}\label{corollary:singledim:verification}
  Let $\automaton = (L, \ell_\init, C, \Sigma, I, E)$ be a TA, $p$ a priority
  function and $\lambda\in\IN\setminus\{0\}$.
  All time-divergent paths of $\automaton$ satisfy the (resp.~direct)
  timed window objective if and only if all time-divergent paths of $\automaton(p, \lambda)$ satisfy the co-Büchi (resp. safety)
  objective over bad locations.
\end{theorem}
\begin{proof}
    We show that if there is a time divergent path of $\automaton$ that violates
  the (resp.~direct) timed window objective, then there is a time-divergent
  path of $\automaton(p, \lambda)$ that violates the co-Büchi (resp.~safety)
  objective for bad locations. The other direction is proven using
  similar arguments and the projection mapping rather than the expansion
  mapping.

  Assume there is a time-divergent initial path
  $\pi = s_0\xrightarrow{d_0, a_0}s_1\ldots$ of
  $\automaton$ that does not satisfy the (resp.~direct) timed window 
  objective.
  Consider the TG $\game = (\automaton, \varnothing, \Sigma)$.
  The sequence $\pi' = s_0((d_0, \bot), (d_0, a_0))s_1\ldots$ is a time-divergent
  initial play of $\game$ as $\pi$ is a time-divergent initial path
  of $\automaton$.
  Furthermore, $\pi'$ does not satisfy the (resp.~direct)
  timed window  objective as it shares the same sequence of states as
  $\pi$. By Lemma~\ref{theorem:window_mapping}, the time-divergent play
  $\ex(\pi')$ of $\game(p, \lambda)$ does not satisfy the
  co-Büchi (resp.~safety)
  objective over bad locations. There is some path of the TA
  $\automaton(p, \lambda)$ that shares the same sequence of states as
  $\ex(\pi')$. This path is time-divergent and does not satisfy the
  co-Büchi (resp.~safety) objective over bad locations, ending the proof.
\end{proof}

\subsection{Translating strategies}\label{section:reduction:translations}
In this section, we present how strategies can be translated from the base
game to the expanded game and vice-versa, using the expansion and projection
mappings. We restrict our attention to move-independent strategies, as this
subclass of strategies suffices for state-based objectives~\cite{AlfaroFHMS03}.

We open the section with a binary classification of time-convergent plays
of the expanded TG, useful to prove how our translations affect blamelessness
of outcomes.
Then we proceed to our translations. We define how a strategy of
$\game(p, \lambda)$ can be translated to a strategy of $\game$ using the
expansion mapping. Then we define a translation of strategies of
$\game$ to strategies of $\game(p, \lambda)$, using the projection mapping.
Each translation definition is accompanied by a technical result that
establishes a connection between outcomes of a translated strategy and
outcomes of the original strategy, through the projection or expansion
mapping. It follows from these technical results that the translation of
a winning strategy in one game is a winning strategy in the other.

Recall that given a winning strategy of $\player_1$, all of its
time-convergent outcomes are $\player_1$-blameless. Therefore, when
translating a winning strategy from one game to the other, we must
argue that all time-convergent outcomes of the obtained strategy
are $\player_1$-blameless.
To argue how our translations preserve this property, let us introduce a
binary classification of time-convergent plays
of $\game(p, \lambda)$. We argue that any time-convergent play of the expanded
TG either remains in a bad location from some point on
or visits finitely many bad locations.
Whenever a bad location is left, a subsequent visit to a bad location requires
at least $\lambda$ time units to elapse. It is thus impossible to visit bad
locations infinitely often without either remaining in a bad location from
some point onward
or having time diverge. We now formalize this property and present its proof.

\begin{proposition}\label{proposition:expansion_structure}
  Let $\bar{\pi} = (\ell_0, q_0, \bar{\nu}_0)(m_0^{(1)}, m_0^{(2)})\cdots$ be
  an infinite play of $\game(p, \lambda)$.
  If $\bar{\pi}$ is time-convergent, then
  $\bar{\pi}\in \cobuchi(L'\setminus \badset)\uplus\cobuchi(\badset)$.
\end{proposition}
\begin{proof}
  Assume towards a contradiction that neither
  $\bar{\pi}\in\cobuchi(L'\setminus \badset)$ nor
  $\bar{\pi}\in \cobuchi(\badset)$ hold. It follows that there are infinitely
  many $q_n$ such that $q_n=\bad$ and infinitely many $j$ such that
  $q_j\neq \bad$. Consider some $n$ such that $q_n=\bad$. There is some
  $j> n$ such that $q_j\neq \bad$ and some $n'> j$ such that
  $q_{n'}=\bad$. We argue that at least $\lambda$ time units pass between steps
  $n$ and $n'$, i.e.,
  $\bar{\nu}_{n'}(\gamma) - \bar{\nu}_n(\gamma) \geq\lambda$.
  Once this is established, iterating this argument establishes divergence
  of $\bar\pi$, reaching a contradiction.
  The invariant of $(\ell_n, \bad)$ ensures that $\bar{\nu}_n(z)=0$.
  The bad location $(\ell_n, q_n)$ is exited at some stage to reach
  $(\ell_j, q_j)$ and the guards of edges to the location $(\ell_{n'}, q_{n'})$
  are $z=\lambda$.
  It follows that at least $\lambda$ time units must elapse between steps
  $n$ and $n'$.
\end{proof}

We now move on to our translations.
We translate $\player_1$-strategies of the expanded TG to
$\player_1$-strategies of the original TG by evaluating the strategy on the
expanded TG on expansions of plays provided by
the expansion mapping and by replacing any occurrences of
$\beta_1$ by $\bot$.
The fact that translating a winning strategy of the expanded TG this way yields
a winning strategy of the original TG is not straightforward: when we
translate a strategy $\bar{\sigma}$ of $\game(p, \lambda)$ to
a strategy $\sigma$ of $\game$, the expansion of an outcome $\pi$ of $\sigma$
may not be consistent with $\bar{\sigma}$, preventing the direct use of
Lemma~\ref{theorem:window_mapping}.
Indeed the definition of the expansion mapping may impose moves
$(d, \beta_1)$ in $\ex(\pi)$ where $\bar{\sigma}$ would suggest $(d, \bot)$.
However, we can mitigate this
issue by constructing another play $\bar{\pi}$ in parallel that is consistent
with $\bar{\sigma}$ and shares the same
sequence of states as $\ex(\pi)$. We leverage the non-deterministic
behavior of tie-breaking and move-independence of $\bar{\sigma}$ to ensure
consistency of $\bar{\pi}$ with $\bar{\sigma}$, by
changing the moves of $\player_2$ on $\bar{\pi}$ comparatively to $\ex(\pi)$.
We also prove that if $\bar{\pi}$ is $\player_1$-blameless and time-convergent,
then $\pi$ also is $\player_1$-blameless.

\begin{lemma}\label{lemma:translation:expansion}
  Let $\bar{\sigma}$ be a move-independent strategy of $\player_1$ in
  $\game(p, \lambda)$.
  Let $\sigma$ be the $\player_1$-strategy in $\game$ defined by
  \[\sigma(\pi) = \begin{cases}
      \bar{\sigma}(\mathsf{Ex}(\pi)) &
      \text{if } \bar{\sigma}(\mathsf{Ex}(\pi))\notin \IR_{\geq 0}\times\{\beta_1\} \\
      (d, \bot) & \text{if } \bar{\sigma}(\mathsf{Ex}(\pi)) = (d, \beta_1)
    \end{cases}\] for all finite plays $\pi\in\plays_\mathit{fin}(\game)$.
  For all $\pi\in\outcome_1(\sigma)$, there is a play
  $\bar{\pi}\in\outcome_1(\bar{\sigma})$ such that $\bar{\pi}$ shares the
  same sequence of states as $\mathsf{Ex}(\pi)$ and such that if $\bar{\pi}$
  is time-convergent and $\player_1$-blameless, then $\pi$ is
  $\player_1$-blameless.
\end{lemma}
\begin{proof}
  Fix an infinite play $\pi\in\outcome_1(\sigma)$. We construct $\bar{\pi}$
  inductively: at step $n\in\IN$ of the construction, we assume that we have
  constructed a finite play
  $\bar{\pi}_n\in\plays_\mathit{fin}(\game(p, \lambda))$ such that
  $\bar{\pi}_n$ shares
  the same sequence of states as $\ex(\pi_{|n})$ and $\bar{\pi}_n$ is
  consistent with $\bar{\sigma}$.

  Initially, we set $\bar{\pi}_0 = \ex(\pi_{|0})$. This is a play consisting of
  only one state and with no moves. It shares the same sequence of states as
  $\ex(\pi_{|0})$ by construction and is consistent with $\bar{\sigma}$.
  
  Now assume by induction that we have constructed a finite play
  $\bar{\pi}_n$ that shares the same sequence of states as $\ex(\pi_{|n})$
  and is consistent with $\bar{\sigma}$.  
  Let $(d^{(1)}, a^{(1)}) = \sigma(\pi_{|n})$.
  Let $m^{(2)}=(d^{(2)}, a^{(2)})$ be a move of $\player_2$ and $s$ be a
  state of $\mathcal{T}(\automaton)$ such that
  $\pi_{|n+1}=\pi_{|n}(\sigma(\pi_{|n}), m^{(2)})s$. Write
  $\last(\bar{\pi}_n)=(\ell, q, \bar{\nu})$.
  We construct $\bar{\pi}_{n+1}$ so that it shares the same sequence of moves
  as $\ex(\pi_{|n+1})$ and so it is consistent with $\bar{\sigma}$.

  Recall that an expansion never ends in a bad location.
  The proof is divided in different cases. Case 1 is when
  $\last(\bar{\pi}_n)$ is in an even location. Case 2 is when
  $\last(\bar{\pi}_n)$ is in an odd location. We further split case 2 in
  four sub-cases:
  \begin{enumerate}[{2}.a]
  \item $\bar{\nu}(z) + d^{(1)}<\lambda$;
  \item $\bar{\nu}(z) + d^{(2)}<\lambda$ and $\bar{\nu}(z) + d^{(1)}=\lambda$;
  \item $d^{(2)} \geq d^{(1)} = \lambda-\bar{\nu}(z)$ and $\player_1$ is
    responsible for the last transition;
  \item $d^{(2)} \geq d^{(1)} = \lambda-\bar{\nu}(z)$ and $\player_1$ is not
    responsible for the last transition (this implies $d^{(2)} = d^{(1)}$).
  \end{enumerate}
  These four sub-cases are disjoint. We show that they cover all
  possibilities before moving on to the remainder of the construction. Assume
   $\last(\bar{\pi}_n)$ is in an odd location.
  The inequality $\bar{\nu}(z) + d^{(1)}\leq\lambda$ holds due to $d^{(1)}$
  being the delay of the move $\sigma(\pi_{|n})$ and this delay being the same
  as the delay of $\bar{\sigma}(\ex(\pi_{|n}))$ (this follows
  from the definition of $\sigma$) and because
  $\last(\ex(\pi_{|n}))=\last(\bar{\pi}_n)$ is an odd location with an
  invariant implying $z\leq\lambda$. We now describe how the construction
  proceeds in each case.
  
  \textbf{Case 1.}
  Assume $\last(\bar{\pi}_n)=\last(\ex(\pi_{|n}))$ is in an even location.
  Then the same moves are
  enabled in $\last(\bar{\pi}_n)$ and
  $\last(\pi_{|n})$.
  We have
  $\ex(\pi_{|n+1}) = \ex(\pi_{|n})(\sigma(\pi_{|n}), m^{(2)})\bar{s}$ for some
  state $\bar{s}$ of $\mathcal{T}(\automaton(p, \lambda))$ given by the
  definition
  of $\ex$. We let $\bar{\pi}_{n+1}$ be
  $\bar{\pi}_{n}(\bar{\sigma}(\bar{\pi}_{n}), m^{(2)})\bar{s}$.
  This is a well-defined play:
  $\ex(\pi_{|n})(\sigma(\pi_{|n}), m^{(2)})\bar{s}$ being a play implies
  $\bar{s}\in\delta(\last(\ex(\pi_{|n})), \sigma(\pi_{|n}), m^{(2)})$ and
  $\bar{\pi}_n$ and $\ex(\pi_{|n})$ share the same sequence of states, hence
  $\last(\bar{\pi}_n) = \last(\ex(\pi_{|n}))$.
  Furthermore, by move-independence of $\bar{\sigma}$ and definition of
  $\sigma$,
  $\sigma(\pi_{|n})=\bar{\sigma}(\ex(\pi_{|n}))=\bar{\sigma}(\bar{\pi}_n)$.
  The play $\bar{\pi}_{n+1}$ is consistent with $\bar{\sigma}$ by construction.

  \textbf{Case 2.}
  Now assume that $\last(\bar{\pi}_n)=(\ell, q, \bar{\nu})$ is in
  an odd location.
  Recall that we must have $\bar{\nu}(z)+d^{(1)}\leq\lambda$.

  \textbf{Sub-case 2.a.}
  Assume $\bar{\nu}(z) + d^{(1)}<\lambda$. Then
  $\ex(\pi_{|n+1}) = \ex(\pi_{|n})(\sigma(\pi_{|n}), \tilde{m}^{(2)})\bar{s}$
  for some state $\bar{s}$ of $\mathcal{T}(\automaton(p, \lambda))$
  given by the definition of the expansion mapping and where
  $\tilde{m}^{(2)} = (\lambda-\bar{\nu}(z), \beta_2)$ if
  $\bar{\nu}(z) + d^{(2)}\geq\lambda$
  and $\tilde{m}^{(2)}=m^{(2)}$ otherwise. We define $\bar{\pi}_{n+1}$ to be
  the play
  $\bar{\pi}_{n}(\bar{\sigma}(\bar{\pi}_{n}), \tilde{m}^{(2)})\bar{s}$, which
  is a well-defined play sharing the sequence of states of $\ex(\pi_{|n+1})$
  and consistent with $\bar{\sigma}$ for the same reasons as case 1.

  \textbf{Sub-case 2.b.} Assume
  $\bar{\nu}(z) + d^{(2)}<\lambda$ and $\bar{\nu}(z) + d^{(1)}=\lambda$. Then
  the move of $\player_1$ is changed in the expansion to $(d^{(1)}, \beta_1)$:
  we have
  $\ex(\pi_{|n+1}) = \ex(\pi_{|n})((d^{(1)}, \beta_1), m^{(2)})\bar{s}$
  for some state $\bar{s}$ of $\mathcal{T}(\automaton(p, \lambda))$.
  It follows from $\player_2$ preempting $\player_1$
  that
  $\last(\ex(\pi_{|n}))\xrightarrow{m^{(2)}}\bar{s}$. By move-independence of
  $\bar{\sigma}$ and definition of $\sigma$,
  the delay of the moves $\sigma(\pi_{|n})$ and
  $\bar{\sigma}(\ex(\pi_{|n})) = \bar{\sigma}(\bar{\pi}_n)$ match. Let
  $\bar{\pi}_{n+1}$ be $\bar{\pi}_n(\bar{\sigma}(\bar{\pi}_n), m^{(2)})\bar{s}$.
  Thus $\bar{\pi}_{n+1}$ is a well-defined play sharing
  the same sequence of states as $\ex(\pi_{|n+1})$ and
  consistent with $\bar{\sigma}$.
  
  \textbf{Sub-cases 2.c.~and 2.d.} For the two remaining cases, assume that
  $\bar{\nu}(z) + d^{(1)}=\lambda$. Thus, $a^{(1)}=\bot$:
  the only $\player_1$ actions available in state $(\ell, q, \bar{\nu}+d^{(1)})$
  are $\beta_1$ and $\bot$ as $I'((\ell, q))$ implies
  $z\leq\lambda$ and edges leaving $(\ell, q)$ with actions other
  than $\beta_1$ and $\beta_2$ have guards requiring $z<\lambda$.
  We write $\bar{\nu}_{d^{(1)}}= \mathsf{reset}_{\{z\}}(\bar{\nu}+d^{(1)})$
  and for $i\in\{1, 2\}$ and $t\geq 0$, we write $b^{(i)}_t = (t, \beta_i)$
  in the following.

  \textbf{Sub-case 2.c.} If $d^{(2)} \geq d^{(1)}$ and $\player_1$ is responsible
  for the last transition of $\pi_{|n+1}$, then
  $\mathsf{Ex}(\pi_{|n+1})$ is of the form
  \[\ex(\pi_{|n})(b^{(1)}_{d^{(1)}}, b^{(2)}_{d^{(1)}})
    (\ell, \bad, \bar{\nu}_{d^{(1)}})
    (b^{(1)}_{0}, b^{(2)}_{0})
    (\ell, p(\ell), \bar{\nu}_{d^{(1)}})((0, \bot), \tilde{m}^{(2)}))
    (\ell, p(\ell), \bar{\nu}_{d^{(1)}})\]
  for some $\tilde{m}^{(2)}$ given by the expansion
  definition.
  We construct $\bar{\pi}_{n+1}$ in three steps.
  The sequence
  $\bar{\pi}_n^+=\bar{\pi}_n(\bar{\sigma}(\bar{\pi}_n), b^{(2)}_{d^{(1)}})
  (\ell, \bad, \bar{\nu}_{d^{(1)}})$ is a well-defined play: both appended
  moves share the same delay because the delay $d^{(1)}$ of the move
  $\sigma(\pi_{|n})$ is that of
  $\bar{\sigma}(\ex(\pi_{|n})) = \bar{\sigma}(\bar{\pi}_n)$ by definition
  of $\sigma$ and move-independence of $\bar{\sigma}$, and
  $\last(\bar{\pi}_n)\xrightarrow{b^{(1)}_{d^{(1)}}}
  (\ell, \bad, \bar{\nu}_{d^{(1)}})$ holds. Its extension
  $\bar{\pi}_n^{++}= \bar{\pi}_n^+(\bar{\sigma}(\bar{\pi}_n^+),b^{(2)}_0)
  (\ell, p(\ell), \bar{\nu}_{d^{(1)}})$ is also a well-defined play as
  $\bar{\sigma}(\bar{\pi}^+)$ must have a delay of zero due to the invariant of
  bad locations enforcing $z=0$. We define $\bar{\pi}_{n+1}$ to be
  the play 
  $\bar{\pi}_n^{++}(\bar{\sigma}(\bar{\pi}_n^{++}),(0, \bot))
  (\ell, p(\ell), \bar{\nu}_{d^{(1)}})$. The sequence $\bar{\pi}_{n+1}$ is
  a play: the move $(0, \bot)$ is available in any state and performing it
  does not change the state, and it cannot be outsped.
  By construction, $\bar{\pi}_{n+1}$ shares the same sequence
  of states as $\mathsf{Ex}(\pi_{|n+1})$ and is consistent with
  $\bar{\sigma}$.

  \textbf{Sub-case 2.d.} Assume now that $d^{(2)}=d^{(1)}$ and that
  $\player_1$ is not responsible for the last transition.
  Then
  $\mathsf{Ex}(\pi_{|n+1})$ is of the form
  \[\ex(\pi_{|n})(b^{(1)}_{d^{(1)}}, b^{(2)}_{d^{(1)}})
    (\ell, \bad, \bar{\nu}_{d^{(1)}})(b^{(1)}_0, b^{(2)}_0)
    (\ell, p(\ell), \bar{\nu}_{d^{(1)}})((0, \bot), (0, a^{(2)}))
    \bar{s}\]
  for some state $\bar{s}$ of $\mathcal{T}(\automaton(p, \lambda))$.
  Like in case 2.c., we extend $\bar{\pi}_n$ using the same states and
  changing the moves along the above to ensure consistency with
  $\bar{\sigma}$. Let $\bar{\pi}_n^{+}$ and $\bar{\pi}_n^{++}$ be
  defined identically to case 2.c.
  We define $\bar{\pi}_{n+1}$ to be
  $\bar{\pi}_n^{++}(\bar{\sigma}(\bar{\pi}_n^{++}),(0, a^{(2)}))\bar{s}$.
  The sequence $\bar{\pi}_{n+1}$ is a well-defined play (the last transition
  depends on $\player_2$'s move), is consistent with $\bar{\sigma}$ and shares
  the same sequence of states as $\mathsf{Ex}(\pi_{|n+1})$.

  The inductive construction above yields a play
  $\bar{\pi}\in\outcome_1(\bar{\sigma})$ that shares the same sequence of
  states as $\ex(\pi)$. It remains to show that if $\bar{\pi}$ is
  time-convergent and $\player_1$-blameless, then $\pi$ is
  $\player_1$-blameless. Assume that $\bar{\pi}$
  is time-convergent and $\player_1$-blameless.
  By Proposition~\ref{proposition:expansion_structure},
  either $\bar{\pi}$ does not leave a bad location from some point on or
  satisfies $\cobuchi(\badset)$. An expansion always exits a bad location
  the step following entry, therefore
  $\bar{\pi}$ is necessarily of the second kind, as it shares its sequence of
  states with $\ex(\pi)$. Thus, from some point on, the inductive construction
  of $\bar{\pi}$ is done using cases 1, 2.a or 2.b. Furthermore, from
  some point on,  $\player_1$'s moves are no longer responsible for
  transitions in $\bar{\pi}$.
  In cases 1 and 2.a, $\player_1$ is responsible for the added transition in
  $\bar{\pi}_n$ if and only if $\player_1$ is responsible for the last
  transition of $\pi_{|n+1}$ (by construction of the expansion mapping). In
  case 2.b, $\player_2$ is strictly faster in both plays. Therefore, if
  from some point on $\player_1$ is no longer responsible for transitions
  in $\bar{\pi}$, then from some point on $\player_1$ is no longer responsible
  for transitions in $\pi$, i.e., $\pi$ is $\player_1$-blameless.
\end{proof}

We can also translate $\player_1$-strategies defined on $\game$ to
$\player_1$-strategies
on $\game(p, \lambda)$ using the projection mapping. Translating strategies
this way must be done with care: we must consider the case where a move
suggested by the strategy in $\game$ requires too long a delay to
be played in the expanded TG $\game(p,\lambda)$. In this case, the suggested
move is replaced by $(d, \beta_1)$ for a suitable delay $d$. By construction,
the translated strategy always suggests the move $(0, \beta_1)$ when the play
ends in a bad location.

Similarly to the first translation, when deriving a strategy $\bar{\sigma}$
of $\game(p, \lambda)$ by translating a strategy $\sigma$ of $\game$,
the projection of an outcome $\bar{\pi}$ of $\bar{\sigma}$ may not be
consistent with $\sigma$.
However, we show that
there is a play $\pi$ consistent with $\sigma$ that shares the
same sequence of states as $\pr(\bar{\pi})$, by using techniques similar
to those used to prove Lemma~\ref{lemma:translation:expansion}.
Analogously to Lemma~\ref{lemma:translation:expansion}, we establish that
time-convergence and $\player_1$-blamelessness of $\pi$ imply
$\player_1$-blamelessness of $\bar{\pi}$.

\begin{lemma}\label{lemma:translation:projection}
  Let $\sigma$ be a move-independent strategy of $\player_1$ in $\game$.
  Let $\bar{\pi}$ be a finite play in $\game(p, \lambda)$ and let
  $(\ell, q, \bar{\nu})$ denote $\last(\bar{\pi})$ and
  $(d^{(1)}, a^{(1)})=\sigma(\pr(\bar{\pi}))$.
  We set:
  \[\bar{\sigma}(\bar{\pi}) = \begin{cases}
      (\lambda-\bar{\nu}(z), \beta_1) & \text{if } q\bmod 2=1 \land \bar{\nu}(z) + d^{(1)}\geq\lambda \\
      (0, \beta_1) & \text{if } q=\bad\\
      \sigma(\mathsf{Pr}(\bar{\pi})) & \text{otherwise}.
    \end{cases}\]
  For all $\bar{\pi}\in\outcome_1(\bar{\sigma})$, there is a play
  $\pi\in\outcome_1(\sigma)$ such that $\pi$ shares the
  same sequence of states as $\pr(\bar{\pi})$ and if $\pi$ is time-convergent
  and $\player_1$-blameless, then $\bar{\pi}$ is $\player_1$-blameless.
\end{lemma}

\begin{proof}
  Fix $\bar{\pi}\in\outcome_1(\bar{\sigma})$. We construct the sought play $\pi$ by
  induction. We assume that at step $n\in\IN$ of the construction,
  we have constructed $\pi_n\in\plays_\mathit{fin}(\game)$ that shares the
  same sequence of states as $\pr(\bar{\pi}_{|n})$, and such that $\pi_n$ is
  consistent with $\sigma$.

  The base case is straightforward. We set $\pi_0=\pr(\bar{\pi}_{|0})$.
  Now, assume that we have constructed a finite play
  $\pi_n\in\plays_\mathit{fin}(\game)$ as described above
  and let us construct $\pi_{n+1}$.
  Let $m^{(2)}=(d^{(2)}, a^{(2)})$ be a move of $\player_2$ and $\bar{s}$
  be a state of $\mathcal{T}(\automaton(p, \lambda))$ such that
  $\bar{\pi}_{|n+1} = \bar{\pi}_{|n}(\bar{\sigma}(\bar{\pi}_{|n}), m^{(2)})\bar{s}$.
  Write $(d^{(1)}, a^{(1)})= \bar{\sigma}(\bar{\pi}_{|n})$,
  $(\ell, q, \bar{\nu})=\last(\bar{\pi}_{|n})$ and
  $(\ell', q', \bar{\nu}')=\bar{s}$.

  We discuss different cases: 1.~$q$ is an even number, 2.~$q=\bad$
  and 3.~$q$ is an odd number. The third case is divided in three disjoint
  sub-cases:
  3.a.~$q'=\bad$; 
  3.b.~$q'\neq \bad$ and $a^{(1)}\neq \beta_1$;
  3.c.~$q'\neq \bad$ and $a^{(1)}=\beta_1$.

  \textbf{Case 1.}
  Assume $q$ is an even number. Then the same moves are enabled
  in $\last(\bar{\pi}_{|n})$
  and $\last(\pr(\bar{\pi}_{|n}))$. By definition of $\bar{\sigma}$ and
  move-independence of $\sigma$,
  $\bar{\sigma}(\bar{\pi}_{|n}) = \sigma(\pr(\bar{\pi}_{|n})) = \sigma(\pi_n)$.
  We have
  $\pr(\bar{\pi}_{|n+1}) =
  \pr(\bar{\pi}_{|n})(\bar{\sigma}(\bar{\pi}_{|n}), m^{(2)})(\ell', \bar{\nu}'_{|C})$.
  We define $\pi_{n+1}$ as the play
  $\pi_n(\sigma(\pi_n), m^{(2)})(\ell', \bar{\nu}'_{|C})$.
  This is a well-defined play consistent with $\sigma$ that shares the
  same sequence of states as
  $\pr(\bar{\pi}_{|n+1})$, due to the
  fact that $\pr(\bar{\pi}_{|n})$ and $\pi_n$ share the same sequence of
  states and
  $\sigma(\pi_n) = \bar{\sigma}(\bar{\pi}_{|n})$.

  \textbf{Case 2.}
  Assume $q=\bad$. We have
  $\ell=\ell'$, $\bar{\nu}=\bar{\nu}'$. Indeed, the only $\player_i$ moves
  available in a bad location are $(0, \bot)$ and $(0, \beta_i)$.
  The play either proceeds using an edge to $(\ell, p(\ell))$ or stays
  in $(\ell, \bad)$. The values of clocks do not change in both of these cases.
  By definition of the projection mapping,
  $\pr(\bar{\pi}_{|n+1}) =
  \pr(\bar{\pi}_{|n})((0, \bot), (0, \bot)) (\ell, \bar{\nu}_{|C})$. We let
  $\pi_{n+1}$ be $\pi_n(\sigma(\pi_n), (0, \bot))(\ell, \bar{\nu}_{|C})$:
  $\pi_{n+1}$ is a well-defined play (the last transition is valid
  by the move $(0, \bot)$ of $\player_2$), is consistent with $\sigma$ and has
  the same sequence of states as
  $\pr(\bar{\pi}_{|n+1})$.

  \textbf{Case 3.}
  Assume $q$ is an odd number. We study three sub-cases:
  a.~$q'=\bad$; b.~$q'\neq \bad$ and $a^{(1)}=\beta_1$; c.~$q'\neq \bad$ and
  $a^{(1)}\neq\beta_1$. They are disjoint and cover all possibilities.
  
  We start with \textbf{case 3.a}. Assume that $q'=\bad$. Then
  $d^{(1)}=d^{(2)}=\lambda-\bar{\nu}(z)$ holds:
  first, to enter a bad location from an odd location, $z$ must reach
  $\lambda$, thus $d^{(i)}\geq \lambda-\bar{\nu}(z)$ for $i\in\{1, 2\}$;
  second, the invariant of odd locations implies $z\leq \lambda$, thus
  no delay greater
  than $\lambda-\bar{\nu}(z)$ can be played by either player in the state
  $\last(\bar{\pi}_n) = (\ell, q, \bar{\nu})$.
  It follows from $d^{(1)}=\lambda-\bar{\nu}(z)$ that
  the delay of the move $\sigma(\pi_n)=\sigma(\pr(\bar{\pi}_{|n}))$
  must be greater than or equal to $d^{(1)}$ (by line 1 of the definition of
  $\bar{\sigma}$).
  We have $\pr(\bar{\pi}_{|n+1}) =
  \pr(\bar{\pi}_{|n})((d^{(1)}, \bot), (d^{(2)}, \bot) (\ell', \bar{\nu}'_{|C})$.
  We define $\pi_{n+1}$ to be
  $\pi_n(\sigma(\pi_n), (d^{(2)}, \bot))(\ell', \bar{\nu}'_{|C})$. The
  sequence $\pi_{n+1}$ is a
  well-defined play consistent with $\sigma$; the last step is valid due to
  the delay of  $\sigma(\pi_{n})$ being more than $d^{(2)}=d^{(1)}$.
  By construction, $\pi_{n+1}$ has the same sequence of states as
  $\pr(\bar{\pi}_{|n+1})$.

  Assume for \textbf{cases 3.b and 3.c} that $q'\neq\bad$.
  We start with \textbf{case 3.b}.
  Assume $a^{(1)}=\beta_1$. Then
  $d^{(1)}=\lambda-\bar{\nu}(z)\geq d^{(2)}$ because $m^{(2)}$ is a
  $\player_2$ move enabled in the state $(\ell, q, \bar{\nu})$ which is
  constrained by the invariant $z\leq\lambda$. As $q'\neq\bad$, necessarily
  $a^{(2)}\neq\beta_2$ and $\player_2$ is responsible for the last transition
  of $\bar{\pi}_{|n+1}$.
  We have
  $\pr(\bar{\pi}_{|n+1})=\pr(\bar{\pi}_{|n})
  ((d^{(1)}, \bot), m^{(2)})(\ell', \bar{\nu}'_{|C})$.
  In this case, we set
  $\pi_{n+1}=\pi_n(\sigma(\pi_n), m^{(2)})(\ell', \bar{\nu}'_{|C})$. This is a
  well-defined play, as $d^{(2)}\leq d^{(1)}$ and $d^{(1)}$ is smaller than the delay
  suggested by $\sigma(\pi_n)=\sigma(\pr(\bar{\pi}_{|n}))$ by definition
  of $\bar{\sigma}$.

  For \textbf{case 3.c}, assume $a^{(1)}\neq \beta_1$.
  Then  $(d^{(1)}, a^{(1)})=\sigma(\pr(\bar{\pi}_{|n}))$ by definition of
  $\sigma$.
  To avoid distinguishing cases following whether $a^{(2)}=\beta_2$,
  let $\tilde{a}^{(2)}$ denote $\bot$ if $a^{(2)}=\beta_2$ and $a^{(2)}$
  otherwise. We have
  $\pr(\bar{\pi}_{|n+1})=
  \pr(\bar{\pi}_{|n})({\sigma}(\pr(\bar{\pi}_{|n})), (d^{(2)}, \tilde{a}^{(2)}))
  (\ell', \bar{\nu}'_{|C})$.
  We let $\pi_{n+1}$ be the sequence
  $\pi_n(\sigma(\pi_n), (d^{(2)}, \tilde{a}^{(2)}))(\ell', \bar{\nu}'_{|C})$.
  It is a well-defined play, consistent with $\sigma$ and
  sharing the same sequence
  of states as $\pr(\bar{\pi}_{|n+1})$.
  These last statements follow from $\pr(\bar{\pi}_{|n})$ and $\pi_n$ sharing
  the same sequence of states and move-independence of $\sigma$ ensuring
  $\sigma(\pr(\bar{\pi}_{|n})) = \sigma(\pi_n)$.

  The previous inductive construction shows the existence of a play
  $\pi\in\outcome_1(\sigma)$ such that $\pr(\bar{\pi})$ and $\pi$
  share the same sequence of states. We now argue that if the constructed play
  $\pi$ is time-convergent and $\player_1$-blameless, then $\bar{\pi}$
  is also $\player_1$-blameless.

  Assume that $\pi$ is time-convergent and $\player_1$-blameless. It follows
  that $\bar{\pi}$ is also time-convergent because $\pi$ and $\pr(\bar{\pi})$
  share the same sequence of states (recall the projection mapping
  preserves time-convergence).
  By Proposition~\ref{proposition:expansion_structure},
  either $\bar{\pi}$ does not leave a bad location from some point on or
  satisfies $\cobuchi(\badset)$. If from some point on $\bar{\pi}$ does not
  leave a bad location, then it is $\player_1$-blameless by construction of
  $\bar{\sigma}$: in a bad location, $\bar{\sigma}$ always suggests a move
  that exits the bad location. Assume now that $\bar{\sigma}$ satisfies
  $\cobuchi(\badset)$. Then, from some point on, the construction of $\pi$
  proceeds following cases 1, 3.b and 3.c of the inductive construction.
  In cases 1 and 3.c, the player responsible for the transition that is
  added is the same in both games. In case 3.b, $\player_1$ is not responsible
  for the last transition of $\bar{\pi}_{|n+1}$. Therefore, if from some point
  on $\player_1$ is no longer responsible for transitions in $\pi$, then
  from some point on, $\player_1$ is no longer responsible for transitions
  in $\bar{\pi}$, i.e., $\bar{\pi}$ is $\player_1$-blameless.
\end{proof}

The translations of strategies described in
Lemma~\ref{lemma:translation:expansion}
(resp.~Lemma~\ref{lemma:translation:projection}) establish a relation associating
to any outcome of a translated strategy, an outcome of the original
strategy which shares states with the expansion (resp. projection) of this
outcome.
Recall that all the objectives we have considered are state-based, and
therefore move-independent strategies suffice for these objectives.
Using the relations developed in
Lemma~\ref{lemma:translation:expansion} and
Lemma~\ref{lemma:translation:projection}, and Lemma~\ref{theorem:window_mapping},
we can establish that translating a winning strategy of $\game$ yields a
winning strategy of $\game(p, \lambda)$, and vice-versa.
From this, we can conclude that the realizability problem on TGs with
(resp.~direct) timed window objectives can be reduced to the realizability problem
on TGs with co-Büchi (resp.~safety) objectives.

\begin{theorem}\label{theorem:singledim:reduction_games}
  Let $s_\init$ be the initial state of $\game$ and $\bar{s}_\init$ be
  the initial state of $\game(p, \lambda)$.
  There is a winning strategy $\sigma$ for  $\player_1$  for the objective
  $\mathsf{TW}(p, \lambda)$ (resp. $\mathsf{DTW}(p, \lambda)$)
  from $s_\init$ in $\game$ if and only if there is a winning strategy
  $\bar{\sigma}$ for $\player_1$ for the objective
  $\cobuchi(\badset)$ (resp. $\safe(\badset)$) from $\bar{s}_\init$ in $\game(p, \lambda)$.
\end{theorem}
\begin{proof}
  Let $(\Psi, \Psi(p, \lambda))\in\{
  (\mathsf{TW}(p, \lambda), \cobuchi(\badset)),
  (\mathsf{DTW}(p, \lambda), \safe(\badset))\}$.

  Assume $\player_1$ has a winning strategy $\sigma$ for the objective $\Psi$ in
  $\game$ from $s_\init$. We assume this strategy to be move-independent
  because $\Psi$ is a state-based objective.
  By Lemma~\ref{lemma:translation:projection}, there is a strategy
  $\bar{\sigma}$ such that for any play
  $\bar{\pi}\in\outcome_1(\bar{\sigma}, \bar{s}_\init)$, there is a play
  $\pi\in\outcome_1(\sigma, s_\init)$ such that the sequence of states
  of $\rho$ and $\pr(\bar{\pi})$ coincide, and if $\pi$ is time-convergent
  and $\player_1$-blameless, then $\bar{\pi}$ is $\player_1$-blameless.

  Fix $\bar{\pi}\in\outcome_1(\bar{\sigma}, \bar{s}_\init)$ and let $\pi$
  as above. We establish that $\bar{\pi}\in\wc_1(\Psi(p, \lambda))$.
  Because $\sigma$ is winning, $\pi\in\wc_1(\Psi)$.
  First, assume $\bar{\pi}$ is time-divergent. Then $\pi$ is
  also time-divergent. It follows that
  $\pi\in\Psi$ and since $\Psi$ is state-based,
  $\pr(\bar{\pi})\in\Psi$. Lemma~\ref{theorem:window_mapping} ensures
  that $\bar{\pi}\in\Psi(p, \lambda)$. Assume now that $\bar{\pi}$ is
  time-convergent. Then $\pi$ is also time-convergent.
  Because $\sigma$ is winning, $\pi$ is $\player_1$-blameless. This implies
  $\player_1$-blamelessness of $\bar{\pi}$.
  Thus, we have $\bar{\pi}\in\wc_1(\Psi(p, \lambda))$ and have shown
  that $\bar{\sigma}$ is winning.

  Conversely, assume $\player_1$ has a winning strategy $\bar{\sigma}$ in
  $\game(p, \lambda)$ for objective $\Psi(p, \lambda)$ from $\bar{s}_\init$.
  We assume that the strategy $\bar{\sigma}$ is move-independent.
  We show that the strategy $\sigma$ defined in
  Lemma~\ref{lemma:translation:expansion} is winning in $\game$ from $s_\init$.
  Fix $\pi\in\outcome_1(\sigma, s_\init)$. Let
  $\bar{\pi}\in\outcome_1(\bar{\sigma}, \bar{s}_\init)$ be such that
  the sequence of states of $\bar{\pi}$ coincides with that of $\ex(\pi)$,
  and that if $\bar{\pi}$ if time-convergent and $\player_1$-blameless,
  then $\pi$ is also blameless. We
  prove that $\pi\in \wc_1(\Psi)$ and distinguish cases following
  time-divergence of $\pi$.
  Assume $\pi$ is time-divergent, then so is $\bar{\pi}$. The play $\bar{\pi}$
  is consistent with a winning strategy and time-divergent, thus
  $\bar{\pi}\in\Psi(p, \lambda)$. Safety/co-Büchi objectives are state-based,
  and hence $\bar{\pi}\in\Psi(p,\lambda)$ if and only if
  $\ex(\pi)\in\Psi(p,\lambda)$.
  It follows from
  Lemma~\ref{theorem:window_mapping} that $\pi\in\Psi$. Now assume
  that $\pi$ is time-convergent. Then, $\bar{\pi}$ is time-convergent thus
  $\player_1$-blameless, therefore $\pi$ is also
  $\player_1$-blameless. We have proven that $\pi\in\wc_1(\Psi)$, which ends the proof.
\end{proof}

\section{Multi-dimensional objectives}\label{section:multi}
We can generalize the former reduction to conjunctions of (resp.~direct)
timed window parity objectives. Fix a TG
$\game=(\automaton, \Sigma_1, \Sigma_2)$
with set of locations $L$. We consider a $k$-dimensional
priority function $p\colon L\to \{0, \ldots, d-1\}^k$. We write
$p_i\colon L\to \{0, \ldots, d-1\}$ for each component function. Fix
$\lambda = (\lambda_1, \ldots, \lambda_k)\in (\IN\setminus\{0\})^k$
a vector of bounds on window sizes.

Generalized (resp.~direct) timed window objectives are defined as
conjunctions of (resp. direct) timed window objectives. Formally, the
\emph{generalized timed window (parity) objective} is defined as
\[\mathsf{GTW}(p, \lambda) =
  \bigcap_{1\leq i\leq k}\mathsf{TW}(p_i, \lambda_i)\]
and the \emph{generalized direct timed window (parity) objective} as 
\[\mathsf{GDTW}(p, \lambda) =
  \bigcap_{1\leq i\leq k}\mathsf{DTW}(p_i, \lambda_i).\]

The verification and realizability problems for these objectives can be
solved using a similar construction to the single-dimensional case.
The inductive property (Lemma~\ref{lemma:inductive_prop})
ensures only one window needs to be monitored at a time on each dimension.
We adapt the construction of the expanded TA to keep track of several
windows at once.

Expanded locations are labeled with vectors $q = (q_1, \ldots, q_k)$ where
$q_i$ represents the smallest priority in the current window on dimension
$i$ for each $i$. Besides these, for each location $\ell$, there is an
expanded location marked by $\bad$.
We do not keep track of the dimension on which a bad window was seen:
we only consider one kind of bad location.
To measure the size of a window in each dimension, we introduce $k$ new clocks
$z_1, \ldots, z_k\notin C$. For any $q\in\{0, \ldots, d-1\}^k$,
denote by $O_q$ the set $\{1\leq i\leq k\mid q_i\bmod 2 =1\}$ of indices of
components that are odd.

Updates of the vector of priorities are independent between dimensions and
follow the same logic as the simpler case. In the single-dimensional case,
an edge leaving an even location $(\ell, q)$ leads to a location of
the form $(\ell', p(\ell'))$ and an edge leaving an odd location $(\ell, q)$
leads to a location of the form $(\ell', \min\{q, p(\ell')\})$.
The behavior of edges is similar in this case. 
The update function $\mathsf{up}$ generalizes the handling of updates.
Let $q$ be a vector of current priorities and $\ell$ be
a location of $\automaton$. We set $\mathsf{up}(q, \ell)$ to be the
vector $q'$ such that $q'_i = \min\{q_i, p_i(\ell)\}$ if $i\in O_q$ and
$q'_i = p_i(\ell)$ otherwise.

For any non-bad location $(\ell, q)$
and edge $(\ell, g, a, D, \ell')$ leaving location $\ell$, there is a
matching edge in the expanded TA. The target location of this edge is
$(\ell', \mathsf{up}(q, \ell'))$. Its guard is a strengthened
version of $g$: it is disabled if a bad window is detected on some dimension by
adding a conjunct $z_i<\lambda_i$ for each $i\in O_q$ to the guard of the edge.
In the same manner that odd locations were equipped with a strengthened
invariant of the form $I(\ell)\land z\leq\lambda$, expanded locations
$(\ell, q)$ are equipped with an invariant that is the conjunction of
$I(\ell)$ with the conjunction over $O_q$ of $z_i\leq\lambda_i$.

It remains to discuss how bad locations are generalized. They have the
invariant $z_1=0$ ($z_1$ is chosen arbitrarily) and for each
$\beta\in\{\beta_1, \beta_2\}$ and location
$\ell$, there is an edge
$((\ell, \bad), \true, \beta, \varnothing, (\ell, p(\ell)))$.
In the single-dimensional case, each odd location had an edge to a
bad location with a guard $z=\lambda$.
Analogously, for any location $(\ell, q)$ with $O_q$ non-empty, there
are edges to
$(\ell, \bad)$. There are two such edges (one for each $\beta\in\{\beta_1, \beta_2\}$) per dimension in $O_q$.
Recall that there cannot be two edges $(\ell, g, a, D, \ell')$
and $(\ell, h, a, D', \ell'')$ with $g\land h$ satisfiable. This prevents
adding $2\cdot |O_q|$ edges guarded by $z_i=\lambda_i$. We do not need more than
$2\cdot |O_q|$ edges however. For instance, instead of having one edge guarded
by $z_1=\lambda_1$ and another guarded by $z_2=\lambda_2$, we can replace
the guard of the second edge by $z_1<\lambda_1 \land z_2=\lambda_2$ to
ensure the guards are incompatible.
Upon entry to a bad location, all additional clocks $z_1$, \ldots, $z_k$
are reset, no matter the dimension on which a bad window was detected.

The formal definition of $\automaton(p, \lambda)$ follows.
\begin{definition}\label{definition:multidim:extended_automaton}
  Given a TA $\automaton = (L, \ell_\init, C, \Sigma, I, E)$,
  $\automaton(p, \lambda)$ is defined to be the
  TA $(L', \ell'_\init, C', \Sigma', I', E')$
  such that
  \begin{itemize}
  \item $L' = L\times (\{0, \ldots, d-1\}^k\cup \{\bad\})$;
  \item $\ell'_\init = (\ell_\init, p(\ell_\init))$;
  \item $C' = C\cup \{z_1, \ldots, z_k\}$ where $z_1, \ldots, z_k\notin C$;
  \item $\Sigma' = \Sigma\cup \{\beta_1, \beta_2\}$ where
    $\beta_1$, $\beta_2\notin\Sigma$;
  \item $I'(\ell, q) = (I(\ell) \land \bigwedge_{i\in O_q}(z_i\leq\lambda_i))$
    for all $\ell\in L$, $q\in\{0, \ldots, d-1\}^k$, and
    $I'(\ell, \bad) = (z_1=0)$ for all $\ell\in L$;
  \item the set of edges $E'$ of $\automaton(p, \lambda)$ is the smallest
    set that satisfies the following rules:
    \begin{itemize}
    \item if $q\neq bad$ and $(\ell, g, a, D, \ell')\in E$, then
      \[\bigg((\ell, q), g\land \bigwedge_{i\in O_q}(z_i< \lambda_i), a,
        D\cup\{z_i\mid i\notin O_q\}, (\ell', \mathsf{up}(q, \ell')\bigg)\in E';\]
    \item for all $\ell\in L$, $q\neq\bad$, $i\in O_q$ and
      $\beta\in\{\beta_1, \beta_2\}$,
      \(((\ell, q), \varphi_i, \beta, \{z_1, \ldots, z_k\}, (\ell, \bad))\in E'\)
      where the clock condition $\varphi_i$ is defined by
      \[\varphi_i \coloneqq (z_i = \lambda_i) \land
        \bigwedge_{\substack{1\leq j < i\\ j\in O_q}} (z_j< \lambda_j);\]
    \item for all $\ell\in L$ and $\beta\in\{\beta_1, \beta_2\}$,
      $((\ell, \bad), \true, \beta, \varnothing, (\ell, p(\ell))\in E'.$
    \end{itemize}
  \end{itemize}

  For a TG $\game = (\automaton, \Sigma_1, \Sigma_2)$,
  we set $\game(p, \lambda) = (\automaton(p, \lambda),
  \Sigma_1\cup\{\beta_1\}, \Sigma_2\cup\{\beta_2\})$.
\end{definition}

The generalized (resp.~direct) timed window objective on a TA
$\automaton$ or TG $\game$ translates to a co-Büchi (resp.~safety) objective on
the expansion $\automaton(p, \lambda)$ and $\game(p, \lambda)$ respectively.
Let $\badset = L\times \{\bad\}$.
We can derive the following theorems by adapting
the proofs of the single-dimensional case
(Theorems~\ref{corollary:singledim:verification}
and~\ref{theorem:singledim:reduction_games}). The only nuance is that there are
several dimensions to be handled in parallel. However, different dimensions
minimally interact with one another. The only way a dimension may affect the
others is by resetting them in the expansion (when a bad location is visited),
and that would entail a bad window was detected on one dimension.

\begin{theorem}\label{theorem:multidim:verification}
  All time-divergent paths of $\automaton$ satisfy the generalized
  (resp.~direct) timed window objective if and only if all time-divergent
  paths of $\automaton(p, \lambda)$ satisfy the co-Büchi (resp. safety)
  objective over bad locations.
\end{theorem}

\begin{theorem}\label{theorem:multidim:games}
  There is a winning strategy $\sigma$ for  $\player_1$ for the objective
  $\mathsf{GTW}(p, \lambda)$ (resp. $\mathsf{GDTW}(p, \lambda)$)
  from the initial state of $\game$ if and only if there is a winning strategy
  $\bar{\sigma}$ for $\player_1$ for the objective $\cobuchi(\badset)$
  (resp. $\safe(\badset)$) from the initial state of $\game(p, \lambda)$.
\end{theorem}

\section{Algorithms and complexity}\label{section:complexity}
This section presents algorithms for solving the verification
and realizability problems for the generalized (resp.~direct)
timed window parity objective. We establish lower bounds that match the complexity class of our algorithms.  We consider the general multi-dimensional setting
and we denote by $k$ the number of timed window
parity objectives under consideration.

\subsection{Algorithms}
We use the construction presented in the previous sections to solve both
the realizability problem and the verification problem for the
(resp.~direct) timed window objective.
In both cases, we invoke known sub-algorithms.
For games, we rely on a general algorithm for
solving TGs with $\omega$-regular location-based objectives.
For automata, we use an algorithm for the emptiness problem of timed automata.

First, we analyze the time required to construct the expanded TA with respect
to the inputs to the problem.
\begin{lemma}\label{lemma:size:reduction}
  Let $\automaton=(L, \ell_\init, C, \Sigma, I, E)$ be a TA.
  Let $p\colon L\to\{0, \ldots, d-1\}^k$ be a $k$-dimensional priority
  function and $\lambda\in (\IN\setminus\{0\})^k$ be a vector of bounds on
  window sizes.
  The expanded TA
  $\automaton(p, \lambda)=(L', \ell_\init', C', \Sigma', I', E')$
  can be computed in time exponential in $k$ and polynomial in the size of
  $L$, the size of $E$, the size of $C$, $d$, the length of the encoding of
  the clock constraints of $\automaton$ and in the encoding of $\lambda$.
\end{lemma}

\begin{proof}
  We analyze each component of $\automaton(p, \lambda)$.
  First, we study the size of the set of expanded locations. The set of expanded
  locations $L'$ is given by $L' = L\times (\{0, \ldots, d-1\}^k\cup \{\bad\})$.
  Therefore, $|L'| = |L|\cdot (d^k+1)$. The set of clocks $C'$ contains
  $|C|+k$ clocks.

  To determine a bound on the size of the set of edges, we base ourselves
  on the rules that define $E'$. For each
  edge $(\ell, g, a, D, \ell')\in E$ and $q\in\{0, \ldots, d-1\}^k$,
  there is an edge in $E'$ exiting location $(\ell, q)$.
  There are
  $|E| \cdot d^k$ edges obtained this way. Furthermore, for each non-bad
  expanded location $(\ell, q)$, there are up to $2k$ edges to the location
  $(\ell, \bad)$. There are at most $2k \cdot |L|\cdot d^k$ edges obtained by this
  rule. Finally, it remains to count the edges that leave bad locations.
  There are two such edges per bad location, totaling to $2 \cdot|L|$.
  Overall, an upper bound on the size of $E'$ is given by
  $|E| \cdot d^k+2k \cdot |L|\cdot d^k+2 \cdot|L|$.

  It remains to study the total length of the encoding of the clock constraints of
  $\automaton(p, \lambda)$. We have shown above that $L'$ and $E'$
  are of size exponential in $k$ and polynomial in $|L|$, $|E|$
  and $d$. There are as many clock constraints
  in $\automaton(p, \lambda)$ as there are locations and edges. Therefore,
  it suffices to show that the encoding of each individual clock constraint
  of $\automaton(p, \lambda)$ is polynomial
  in the total length of the encoding of the clock constraints of $\automaton$,
  $k$ and the encoding of $\lambda$ to end the proof.
  
  A clock constraint of $\automaton(p, \lambda)$
  is either derived from a clock constraint of $\automaton$,
  the invariant of a bad location or the guard of
  an edge to or from a bad location.
  Clock constraints derived from $\automaton$ are either unchanged, or they
  are obtained by reinforcing a clock constraint of $\automaton$ with
  conjuncts $z_i\leq\lambda_i$ in the case of invariants or
  with a conjuncts $z_i<\lambda_i$ in the case of guards. At most, we extend
  clock constraints of $\automaton$ with $k$ conjuncts that can be encoded
  linearly w.r.t.~the encoding of the $\lambda_i$. The invariant of bad
  locations is $z_1=0$ and is therefore constant in length. The guards of
  edges exiting bad locations are $\true$ and are also constant in length.
  Finally, guards of edges to bad locations are conjunctions of the form
  $z_{i_1}<\lambda_{i_1}\land\ldots\land z_{i_{j-1}}<\lambda_{i_{j-1}}\land
  z_{i_j}=\lambda_{i_j}$, which can be encoded linearly w.r.t.~$k$ and
  the encoding of the $\lambda_i$.
\end{proof}

\subparagraph{Timed games.} To solve the realizability problem, we rely on
the algorithm from~\cite{AlfaroFHMS03}.
To use the algorithm of~\cite{AlfaroFHMS03}, we introduce
deterministic parity automata. For all finite alphabets $A$,
deterministic parity automata can represent all $\omega$-regular
subsets of $A^\omega$. We use deterministic parity automata to encode
location-based objectives (the set of locations of the studied timed automaton
serves as the alphabet of the parity automaton).

Let $A$ be a finite non-empty alphabet.
A \textit{deterministic parity automaton} (DPA) of order $m$ is a tuple
$H = (Q, q_\init, A, \delta, \Omega)$, where $Q$ is a finite set of states,
$q_\init\in Q$ is the initial state, $\delta\colon Q\times A\to Q$ is the
transition function and $\Omega: Q\to \{0, \ldots, 2m-1\}$ is a function assigning
a priority to each state of the DPA. An execution of $H$ on an infinite
word $a_0a_1\ldots\in A^\omega$ is an infinite sequence of states
$q_0q_1\ldots\in Q^\omega$ that starts in the initial state of $H$, i.e.,
$q_0=q_\init$ and such that for all $i\in\IN$, $q_{i+1}=\delta(q_i, a_i)$,
i.e., each step of
the execution is performed by reading a letter of the input word.
An infinite word $w\in A^\omega$ is accepted by $H$ if there is an execution
$q_0q_1\ldots\in Q^\omega$ such that the smallest priority appearing
infinitely often along the execution is even, i.e. if
$(\liminf_{n\to\infty}\Omega(q_i))\bmod 2=0$. A DPA
is \textit{total} if its transition function is total, i.e., for all $q\in Q$
and $a\in A$, $\delta(q, a)$ is defined.

The algorithm of~\cite{AlfaroFHMS03} checks the existence of a
$\player_1$-winning strategy in a TG
$\game=(\automaton, \Sigma_1, \Sigma_2)$ where the set of locations of
$\automaton$ is $L$ and the set of clocks is $C$
with a location objective specified by a total DPA
$H$ with set of states $Q$ and of order $m$ in time
\begin{equation}\label{equation:complexity:games}
  \mathcal{O}\left(\left(|L|^2\cdot |C|!\cdot
      2^{|C|}\cdot \prod_{x\in C}(2c_x+1) \cdot |Q| \cdot m\right)^{m+2}\right),
\end{equation}
where $c_x$ is the largest constant to which $x$ is compared to in the
clock constraints of $\automaton$.

We now show that the realizability problem for generalized (resp.~direct)
timed window objectives is in $\textsf{EXPTIME}$ using the algorithm
described above. This essentially boils down to specifying DPAs that
encode safety and co-Büchi objectives and using these along with the
algorithm of~\cite{AlfaroFHMS03} to check the existence of a
$\player_1$-winning strategy in $\game(p, \lambda)$.

\begin{proof}
Fix a TG $\game = (\automaton, \Sigma_1, \Sigma_2)$ and let $L$ denote the
set of locations of $\automaton$. Let $p$ be a multi-dimensional priority
function and $\lambda\in (\IN\setminus\{0\})^k$ be a vector of window sizes.
Let $L'$ denote the set of locations of $\automaton(p, \lambda)$.
We describe total DPAs
$H_{\safe(\badset)}$ and $H_{\cobuchi(\badset)}$ for the objectives $\safe(\badset)$
and $\cobuchi(\badset)$ in the TG $\game(p, \lambda)$.

We encode the safety objective using a DPA with two states. Intuitively,
the initial state is never left as long as no bad location is read. If the
DPA reads a bad location, it moves to a
sink state, representing that the safety objective was violated.
The initial state is given an even priority
and the sink state an odd priority so that accepting executions are those
that read sequences of locations that respect the safety objective.
Formally, let $H_{\safe(\badset)}$ be the DPA
$(Q, q_\init, L', \delta, \Omega)$ where $Q = \{q_\init, q_\bad\}$,
$\Omega(q_\init) = 0$ and $\Omega(q_\bad) = 1$, and the transition function
is defined by $\delta(q_\init, (\ell, q)) = q_\init$ for all $(\ell, q)\in L'\setminus \badset$, $\delta(q_\init, (\ell, \bad)) = q_\bad$ for all $\ell\in L$
and $\delta(q_\bad, (\ell, q))=q_\bad$ for all $(\ell, q)\in L'$. This
DPA is total by construction.

The co-Büchi objective is also encoded by a DPA with two states.
The first state $q_\init$ is entered every time a non-bad location is read
by the DPA and the second state $q_\bad$ whenever a bad location
is read. Runs that visit $q_\bad$ infinitely often violate the co-Büchi
objective, therefore we give $q_\bad$ an odd priority smaller than the
even priority of $q_\init$.
Formally, let $H_{\cobuchi(\badset)}$ be the DPA
$(Q, q_\init, L', \delta, \Omega)$ where $Q = \{q_\init, q_\bad\}$,
$\Omega(q_\init) = 2$ and $\Omega(q_\bad) = 1$, and the transition function
is defined by $\delta(q, (\ell, q)) = q_\init$ for all $(\ell, q)\in L'\setminus \badset$ and $q\in\{q_\init, q_\bad\}$, and
$\delta(q, (\ell, \bad)) = q_\bad$ for all $\ell\in L$ and
$q\in\{q_\init, q_\bad\}$.
This deterministic parity automaton is total by construction.

By Theorem~\ref{theorem:multidim:games}, it suffices to check the existence of
a winning strategy for $\player_1$ in the expanded TG $\game(p, \lambda)$ to
answer the realizability problem over $\game$.
By Lemma~\ref{lemma:size:reduction}, the construction of the TG
$\game(p, \lambda)$ from $\game$ takes exponential time w.r.t.~to the
inputs to the problem.  

Recall that there are
$|L|\cdot (d^k+1)$  locations in $\automaton(p, \lambda)$.
We have shown that the objective $\safe(\badset)$ (resp.~$\cobuchi(\badset)$)
can be encoded using a DPA with two states and order at most $2$.
Combining this with
equation \eqref{equation:complexity:games}, it follows that we
can check the existence of a $\player_1$-winning strategy in $\game(p, \lambda)$
for a safety or co-Büchi objective in time
\begin{equation}\label{equation:complexity:window_games}
  \mathcal{O}\left(\left(
      |L|^2\cdot (d^k+1)^2 \cdot (|C| + k)!\cdot 2^{|C|+k}
      \prod_{x\in C}(2c_x+1) \cdot \prod_{1\leq i\leq k}(2\lambda_i+1)
    \right)^4\right),
\end{equation}
where $c_x$ is the largest constant to which $x$ is compared to in the
clock constraints of $\automaton$.

Overall, $\game(p, \lambda)$ can be constructed in exponential time
and the existence of
a $\player_1$-winning strategy in $\game(p, \lambda)$ can be checked in exponential time,
establishing \textsf{EXPTIME}-membership of the realizability problem for
(resp.~direct) timed window objectives.
\end{proof}

\subparagraph{Timed automata.} We describe a
polynomial space algorithm for the verification problem for
generalized (resp.~direct) timed window objectives derived from our reduction.
First, we remark that to check that all time-divergent paths of a TA
satisfy a conjunction of objectives, we can proceed one objective at a time:
for any family of objectives, there is some time-divergent path that
does not satisfy the conjunction of objectives in the family if and only if,
for some objective from the family, there is some time-divergent path that does
not satisfy it. This contrasts with TGs, in which $\player_1$ may have a
winning strategy for each individual objective but is unable to satisfy their
conjunction.

Following the observation above, we establish membership in
\textsf{PSPACE} of the verification problem for the generalized
(resp.~direct) timed window objective in two steps. First, we argue that the
verification problem is in \textsf{PSPACE} for the single-dimensional (resp.~direct) timed window parity objective.
Then, in the multi-dimensional setting,
we use the single-dimensional case as an oracle to check satisfaction of
a generalized objective one dimension at a time.

\begin{lemma}\label{lemma:complexity:automata}
  The generalized (direct) direct timed window verification problem
  is in \textsf{PSPACE}.
\end{lemma}
\begin{proof}

Fix a TA $\automaton=(L, \ell_\init, C, \Sigma, I, E)$, a
priority function $p\colon L\to \{0, \ldots, d-1\}$ and
a bound $\lambda\in \IN\setminus\{0\}$ on window sizes. 
By Theorem~\ref{corollary:singledim:verification}, the verification
problem on $\automaton$ for the (resp.~direct) timed window objective
can be reduced to the verification problem on
$\automaton(p, \lambda)$ for the co-Büchi objective
$\cobuchi(\badset)$ (resp.~the safety objective $\safe(\badset)$).
The complement of a co-Büchi (resp.~safety) objective is a Büchi
(resp.~reachability) objective.
Therefore, there is a time-divergent path of $\automaton$ that does not satisfy the (resp.~direct) timed window objective if and
only if there is some time-divergent path in $\automaton(p, \lambda)$
that satisfies $\buchi(\badset)$ (resp.~$\reach(\badset)$).

Our algorithm for the verification problem on $\automaton$ for the
(resp.~direct) timed window objective proceeds as follows: construct
$\automaton(p, \lambda)$ and check if there is some time-divergent path
in $\automaton(p, \lambda)$ satisfying $\buchi(\badset)$
(resp.~$\reach(\badset)$) and return no if that is the case.

This algorithm is in polynomial space.  By
Lemma~\ref{lemma:size:reduction}, as we have fixed $k=1$,
$\automaton(p, \lambda)$ can be computed in time polynomial in $d$,
the sizes of $L$, $E$, $C$, the length of the encoding of the clock constraints
of $\automaton$
and the encoding of $\lambda$. In other words, the verification problem
for the (resp.~direct) timed window objective can be reduced to
checking the existence of a time-divergent path of a TA satisfying a Büchi
(resp.~reachability objective) in polynomial time.
Checking the existence of a time-divergent path satisfying a Büchi or
reachability objective can be done in polynomial space~\cite{AlurD94}.
Thus, the verification problem for the (one-dimensional)
(resp.~direct) timed window objective can be solved in polynomial space.

For the $k$-dimensional case, the previous algorithm can be used for
each individual component, to check that all dimensions satisfy
their respective objective.
Complexity-wise, this shows the multidimensional
problem belongs in $\mathsf{P}^{\mathsf{PSPACE}}$ (with an oracle in
\textsf{PSPACE} for the single-dimensional case).
Because, $\mathsf{P}^{\mathsf{PSPACE}} = \mathsf{PSPACE}$~\cite{DBLP:journals/siamcomp/BakerGS75}, this proves that the generalized (resp.~direct) timed window
verification problem is in \textsf{PSPACE}.
\end{proof}

\subsection{Lower bounds}
We have presented algorithms which share the same complexity class for one or
multiple dimensions. In this section, we establish that our bounds are tight.
It suffices to show hardness for the single-dimensional problems,
as they are subsumed by the $k$-dimensional case.

The verification and realizability problems for timed window
objectives can be shown to be at least as hard as the verification and
realizability problems for safety objectives. The safety verification problem is \textsf{PSPACE}-complete~\cite{AlurD94}.
The realizability problem for safety objectives is
\textsf{EXPTIME}-complete (this follows from \textsf{EXPTIME}-completeness
of the safety control problem~\cite{HenzingerK99}).
The same construction is used
for the verification and realizability problem. Given a timed automaton,
we expand it so as to encode in locations whether the safety objective was
violated at some point. We assign an even priority to locations that indicate
the safety objective never was violated and an odd priority to the other locations:
as long as the safety objective is not violated, windows close immediately and
as soon as it is violated, it no longer is possible to close windows.

\begin{lemma}\label{lemma:complexity:lowerbound}
  The verification (resp.~realizability) problem for the (direct)
  timed window objective is \textsf{PSPACE}-hard
  (resp.~\textsf{EXPTIME}-hard).
\end{lemma}
\begin{proof}
Fix a TA $\automaton=(L, \ell_\init, C, \Sigma, I, E)$ and a set of
unsafe locations $F\subseteq L$. We construct a TA $\automaton'$ and a
priority function $p$ such that all time-divergent initial paths of $\automaton$
satisfy $\safe(F)$ if and only if all time-divergent initial paths of
$\automaton'$ satisfy $\mathsf{TW}(p, 1)$ (resp.~$\mathsf{DTW}(p, 1)$) (the
choice of 1 as the bound of the window size is arbitrary, the provided
reduction functions for any bound on the size of windows).
We encode the safety objective in a TA $\automaton'$ as
a (resp.~direct) timed window objective.

We expand locations of $\automaton$ with a Boolean
representing whether $F$ was visited.
Formally, let $\automaton'=(L', \ell_\init', C, \Sigma, I', E')$ be the TA
with $L' = L\times\{0, 1\}$, $\ell'_\init=(\ell_\init, 0)$ if
$\ell_\init\notin F$ and $\ell'_\init=(\ell_\init, 1)$ otherwise,
$I'((\ell, b))=I(\ell)$ for all $\ell\in L$ and $b\in\{0, 1\}$, and the
set of edges $E'$ is the smallest set satisfying, for each edge
$(\ell, g, a, D, \ell')\in E$: $((\ell, 0), g, a, D, (\ell', 0))\in E'$
if $\ell'\notin F$; $((\ell, 0), g, a, D, (\ell', 1))\in E'$ if
$\ell'\in F$; and $((\ell, 1), g, a, D, (\ell', 1))\in E'$.
The priority function $p$ defined over locations of $\automaton'$ is defined
by $p((\ell, b))=b$.

There is a natural bijection $f$ between the set of initial paths
of $\automaton$ and the set of initial paths of  $\automaton'$. An initial path
$\pi=(\ell_0, \nu_0)\xrightarrow{m_0}(\ell_1, \nu_1)\ldots\paths(\automaton)$
is mapped via $f$ to the initial path
$\bar{\pi}=((\ell_0, b_0), \nu_0)\xrightarrow{m_0}((\ell_1, b_1), \nu_1)\ldots$
of $\automaton'$, where the sequence $(b_n)_{n\in\IN}$ is $(0)_{n\in\IN}$ if
for all $n\in\IN$, $\ell_n\notin F$ and otherwise,
$b_n=0$ for all $n<n^\star$ and $b_n=1$ for all $n\geq n^\star$,
where $n^\star$ denotes $\min\{n\in\IN\mid \ell_n\in F\}$. This mapping
is well-defined and bijective:
the same moves are enabled in a state $(\ell, \nu)$ of
$\mathcal{T}(\automaton)$ and in the state $((\ell, b), \nu)$ of $\mathcal{T}(\automaton')$
for all
$b\in\{0, 1\}$.
Furthermore, for all $b\in\{0, 1\}$ and all moves $m$ enabled in $(\ell, \nu)$,
$(\ell, \nu)\xrightarrow{m}(\ell', \nu')$ holds if and only if there
is some $b'$ such that $((\ell, b), \nu)\xrightarrow{m}((\ell', b'), \nu')$

The bijection $f$ preserves time-divergence. Furthermore, a path $\pi$ of $\automaton$ satisfies
$\safe(F)$ if and only if $f(\pi)$ does not visit a location of the form
$(\ell, 1)$. The initial paths of $\automaton'$ that visit a location of the form
$(\ell, 1)$
are exactly those that do not satisfy the (resp.~direct) timed window objective:
once such a location is entered, it is no longer possible to close
windows as the set of locations $L\times\{1\}$ cannot be left by construction.
Therefore, there is a time-divergent path of $\automaton$
that does not satisfy $\safe(F)$ if and only if there is a time-divergent
path of $\automaton'$ that does not satisfy $\mathsf{TW}(p, 1)$
(resp.~$\mathsf{DTW}(p, 1)$). Furthermore, the TA $\automaton'$ has the same
set of clocks as $\automaton$ and twice as many locations and edges as
$\automaton$, and the overall length of the encoding of the clock constraints
of $\automaton'$ is double that of $\automaton$.
This shows that the safety verification
problem can be reduced in polynomial time to the (resp.~direct) timed window
objective. This establishes \textsf{PSPACE}-hardness of the verification
problem for (resp.~direct) timed window objectives.

The same construction can be used for the realizability problem. There
is an analogous mapping for initial plays. This mapping can be used to
establish a bijection between the restriction of strategies over initial
paths in the safety TG and the (resp.~direct) timed window TG. It follows that
the realizability problem for safety objectives can be reduced to the
realizability problem for (resp.~direct) timed window objectives in polynomial
time, establishing \textsf{EXPTIME}-hardness of the realizability
problem for (resp.~direct) timed window objectives.
\end{proof}

\subsection{Wrap-up}
We summarize our complexity results in the following theorem.
\begin{theorem}\label{theorem:complexity}
  The verification problem for the (direct) generalized timed window parity
  objective is \textsf{PSPACE}-complete and the realizability problem
  for the (direct) generalized timed window parity objective is
  \textsf{EXPTIME}-complete.
\end{theorem}

We conclude with a comparison of TGs with parity objectives and TGs with
(resp.~direct) timed window objectives. It was shown in~\cite{ChatterjeeHP11}
that TGs with parity objectives can be reduced to turn-based parity
games on graphs. Their solution is as follows: to check if a $\player_1$-winning
strategy exists in a TG $\game=(\automaton, \Sigma_1, \Sigma_2)$ with the
objective $\mathsf{Parity}(p)$, they construct a turn-based game on a graph
with $256\cdot |S_{\mathsf{Reg}}|\cdot |C|\cdot d$ states and $d+2$ priorities,
where $S_{\mathsf{Reg}}$ denotes the set of states of the region abstraction of
$\automaton$, the size of which is bounded by
$|L|\cdot |C|!\cdot 2^{|C|}\cdot \prod_{x\in C}(2c_x+1)$.
Despite recent progress on quasi-polynomial-time algorithms for parity games~\cite{CaludeJKL017}, there are no known polynomial-time algorithms, and, in many techniques, the blow-up is due to the number of priorities. Therefore, the complexity of
checking the existence of a $\player_1$-winning strategy in the TG $\game$ for
a parity objective suffers from a blow-up related to the number of
priorities and the size of the region abstraction.

In contrast, our solution for TGs with (resp.~direct) timed window objectives
does not
suffer from a blow-up due to the number of priorities: for a single dimension,
the complexity given in \eqref{equation:complexity:window_games} is polynomial
in the size of the set of states of the region abstraction and in $\lambda$.
This shows that adding time bounds to parity objectives in TGs comes for free
complexity-wise.

\bibliography{bib}

\end{document}